\documentclass[letterpaper,twocolumn,10pt,accepted=2024-05-29]{quantumarticle}

\pdfoutput=1
\usepackage[utf8]{inputenc}
\usepackage[english]{babel}
\usepackage[T1]{fontenc}
\usepackage{amsmath}
\usepackage{hyperref}
\usepackage{mathtools}
\usepackage{amssymb}
\usepackage{graphicx}
\usepackage{bm}
\usepackage{subfigure}
\usepackage{array}
\usepackage{multirow}
\newtheorem{theorem}{Theorem}
\newtheorem{definition}[theorem]{Definition}
\newtheorem{lemma}[theorem]{Lemma}
\newtheorem{problem}[theorem]{Problem}
\newtheorem{proposition}[theorem]{Proposition}
\newtheorem{remark}[theorem]{Remark}
\newcommand{\qipeb}{QIP$_{\operatorname{EB}}(2)$}

\newenvironment{proof}[1][Proof]{\noindent\textbf{#1.} }{\ \rule{0.5em}{0.5em}}

\allowdisplaybreaks

\begin{document}

\title{\texorpdfstring{Schr\"odinger}{Schrodinger} as a Quantum Programmer:\newline
Estimating Entanglement via Steering}

\author{Aby Philip}
\orcid{0000-0002-4608-7522}
\author{Soorya Rethinasamy}
\orcid{0000-0002-8849-3748}
\affiliation{School of Applied and Engineering Physics,
Cornell University, Ithaca, New York 14850, USA}

\author{Vincent Russo}
\affiliation{Unitary Fund}

\author{Mark M. Wilde}
\affiliation{School of Electrical and Computer Engineering,
Cornell University, Ithaca, New York 14850, USA}
\affiliation{School of Applied and Engineering Physics,
Cornell University, Ithaca, New York 14850, USA}
\orcid{0000-0002-3916-4462}

\begin{abstract}
    Quantifying entanglement is an important task by which the resourcefulness of a quantum state can be measured. Here, we develop a quantum algorithm that tests for and quantifies the separability of a general bipartite state by using the quantum steering effect, the latter initially discovered by Schr\"odinger. Our separability test consists of a distributed quantum computation involving two parties: a computationally limited client, who prepares a purification of the state of interest, and a computationally unbounded server, who tries to steer the reduced systems to a probabilistic ensemble of pure product states. To design a practical algorithm, we replace the role of the server with a combination of parameterized unitary circuits and classical optimization techniques to perform the necessary computation. The result is a variational quantum steering algorithm (VQSA), a modified separability test that is implementable on quantum computers that are available today. We then simulate our VQSA on noisy quantum simulators and find favorable convergence properties on the examples tested. We also develop semidefinite programs, executable on classical computers, that benchmark the results obtained from our VQSA. Thus, our findings provide a meaningful connection between steering, entanglement, quantum algorithms, and quantum computational complexity theory. They also demonstrate the value of a parameterized mid-circuit measurement in a VQSA.
\end{abstract}

\maketitle
\tableofcontents

\section{Introduction}
    Entanglement is a unique feature of quantum mechanics, initially brought to light by Einstein, Podolsky, and Rosen~\cite{Einstein1935}. Many years later, the modern definition of entanglement was given~\cite{W89}, which we recall now.
    A bipartite quantum state $\sigma_{AB}$ of two spatially separated systems $A$ and $B$ is separable (unentangled) if it can be written as a probabilistic mixture of product states~\cite{W89}:
    \begin{equation}
    \label{eqn:sep-state-informal}
    \sigma_{AB}= \sum_{x \in \mathcal{X}} p(x)\, \psi^{x}_{A} \otimes \phi^{x}_{B}    ,
    \end{equation}
     where $\{p (x) \}_{x\in \mathcal{X}}$ is a probability distribution and $\psi^{x}_{A}$ and $\phi^{x}_{B} $ are pure states. The idea here is that the correlations between~$A$ and $B$ can be fully attributed to a classical, inaccessible random variable with probability distribution $\{p (x)\}_{x\in \mathcal{X}}$.
      
     The definition above is straightforward to write down, but it is a different matter to formulate an algorithm to decide if a general state is separable; in fact, it has been proven to be computationally difficult in a variety of frameworks~\cite{G03, Gharibian2010, HMW13, HMW14, GHMW15}. Intuitively, deciding the answer requires performing a search over all possible probabilistic decompositions of the state, and there are too many possibilities to consider. Regardless, determining whether a general state $\rho_{AB}$ is separable or entangled, known as the separability problem, is a fundamental problem of interest relevant to various fields of physics, including condensed matter~\cite{RevModPhys.80.517,cramer2011measuring, laflorencie2016quantum},  quantum gravity~\cite{Takayanagi_2012,bose2017spin,marletto2017gravitationally, Qi2018, Swingle18}, quantum optics~\cite{RevModPhys.92.035005}, and quantum key distribution~\cite{Ekert91, PhysRevLett.113.140501}. In quantum information science, entanglement is the core resource in several basic quantum information processing tasks~\cite{Ekert91, bennett1992communication, bennett1993teleporting}, making the separability problem essential in this field as well.

    Part of the challenge in using entangled states for various tasks is that they are hard to produce and maintain faithfully on any physical platform. The utility of entangled states drops off dramatically the further they are from being perfectly or maximally entangled. Therefore, assessing the quality of entangled states produced becomes an important task, thus motivating the problem of quantifying entanglement~\cite{BDSW96, VPRK97, VP98, Horodecki2009}, in addition to deciding whether entanglement is present.

    To check whether a state is entangled and to quantify its entanglement content experimentally, a rudimentary approach employs state tomography to reconstruct the density matrix and check whether the matrix represents a state that is entangled~\cite{Home2006, Steffen2006}. However, the computational complexity of this method scales exponentially with the number of qubits, thus prohibiting its use on larger states of interest. With the rapid development of quantum computers of increasing size, it is already infeasible to perform tomography to estimate the density matrices describing the states of these computers. It is even more daunting to address the separability problem using various well-known one-sided entanglement tests~\cite{Peres1996, Horodecki1996, W89a, Doherty2004}. This leaves us to seek out alternative methods for addressing the separability problem, and one forward-thinking direction is to employ a quantum computer to do so~\cite{HMW13, HMW14, GHMW15, LRW21}. 

    An approach for addressing the separability problem, which we employ here, involves the quantum steering effect, originally discovered by Schr\"odinger~\cite{Schrodinger1935, Schroedinger1935}. The idea of steering is that if two distant systems are entangled, distinct probabilistic ensembles of states can be prepared on one system by performing distinct measurements on the other system. To describe this phenomenon more precisely, we can employ some elementary notions from quantum mechanics. Let $\psi_{CD}$ be a pure state of two distant quantum systems~$C$ and~$D$, and let $\rho_C = \operatorname{Tr}_D[\psi_{CD}]$ be the reduced state of the system~$C$. Then by performing a measurement on the system $D$, it is possible to realize a probabilistic ensemble $\{(p (z),\psi^z_C)\}_z$ of pure states on the system $C$ that satisfies $\rho_C = \sum_z p (z)\psi^z_C$. Moreover, for each possible probabilistic decomposition of $\rho_C$, a measurement acting on~$D$ can realize this decomposition. Steering has been a topic of interest in recent years, with applications to quantum key distribution~\cite{CS17, UCNG20}, quantum optics~\cite{PhysRevLett.128.200401, PhysRevA.106.042414}, and the foundations of quantum mechanics~\cite{wittmann2012loophole, PhysRevA.91.012112}.

    As suggested above, we can make a non-trivial link between the separability problem and steering, which offers a quantum mechanical method for approaching the former. To see it, recall that
    a purification of the separable state $\sigma_{AB}$ in~\eqref{eqn:sep-state-informal} 
      is a pure state $\varphi_{RAB} $ that satisfies $\operatorname{Tr}_R[\varphi_{RAB}] = \sigma_{AB}$, and consider that one such  choice of the state vector $|\varphi\rangle_{RAB}$ in this case is as follows:
    \begin{equation}
        |\varphi\rangle_{RAB}=\sum_{x\in \mathcal{X}} \sqrt{p (x)}\,|x\rangle_R\otimes |\psi^{x}\rangle_{A} \otimes |\phi^{x}\rangle_{B},
        \label{eq:purif-sep-state}
    \end{equation}
    where $\{ |x\rangle_R\}_{x\in \mathcal{X}}$ is an orthonormal basis. Purifications are not unique, but all other purifications of $\sigma_{AB}$ are related to the one in~\eqref{eq:purif-sep-state} by the action of a unitary operation on the reference system $R$~\cite{NC00}.
    By inspecting~\eqref{eq:purif-sep-state}, we see that the systems $A$ and $B$ can be steered into the probabilistic ensemble $\{(p(x),\psi^{x}_{A} \otimes \phi^{x}_{B})\}_{x\in \mathcal{X}}$ of product states by performing the projective measurement $\{|x\rangle\!\langle x|_R\}_{x\in \mathcal{X}}$ on the reference system~$R$ of $\varphi_{RAB}$. This leads to an idea for testing separability in the general case. If a purification of a general state $\rho_{AB}$ is available and the state $\rho_{AB}$ is indeed separable, then one can a) try to find the unitary that realizes the purification in~\eqref{eq:purif-sep-state} and b) perform the measurement $\{|x\rangle\!\langle x|_R\}_{x\in \mathcal{X}}$ on the reference system~$R$. After receiving the outcome~$x$,  one can finally test whether the reduced state is a product state.
    
    As we will see in more detail later, the basic idea outlined above is at the heart of our method to test whether a state is separable. Additionally, this approach leads to a quantum algorithm and complexity-theoretic statements for quantifying the amount of entanglement in a state. We thus provide a meaningful connection between steering, entanglement, quantum algorithms, and quantum computational complexity theory, which has not been observed hitherto. 

In this paper, we expand on the abovementioned idea to develop various 
separability tests using the quantum steering effect. Our separability test for mixed states consists of a distributed quantum computation involving two parties: a computationally unbounded server, called a prover, which can, in principle, perform any quantum computation imaginable, and a computationally limited client, called a verifier, which can perform time-efficient quantum computations (see Figure~\ref{fig:Max_Sep_Fidelity_QIP_EB}). We prove Theorems~\ref{theorem:qip-eb_msf} and~\ref{theorem:global_cost_func}, which state that the acceptance probabilities of our algorithms, in the ideal case, are directly related to a bonafide entanglement measure, the fidelity of separability. We also employ concepts from quantum computational complexity theory~\cite{watrous2009complexity, VW15} to understand how difficult this test is to perform. Our second contribution results from a modification of our separability test. In an attempt to design a practical algorithm, we replace the prover with a combination of parameterized unitary circuits and classical optimization techniques to perform the necessary computation.  This results in a variational quantum steering algorithm (VQSA) that approximates the aforementioned separability test (see Figure~\ref{fig:Max_Sep_Fidelity_VQA}). The concept of quantum steering is again at the heart of our VQSA, just like the test for separability that it approximates. Interestingly, we prove that the acceptance probability of both tests is related to an entanglement measure called fidelity of separability~\cite{VPRK97, VP98}. We also generalize our separability test and VQSA to the multipartite setting using appropriate definitions of multipartite separability. 
    
    Next, we report the results of simulations of the VQSA on a quantum simulator and find that they show favorable convergence properties. In light of the limited scale and error tolerance of near-term quantum computers, we develop semidefinite programs (SDPs) to approximate the fidelity of separability using positive-partial-transpose (PPT) conditions~\cite{Peres1996, Horodecki1996} and $k$-extendibility~\cite{W89a, Doherty2004} to benchmark the results obtained from our VQSA. As variational quantum algorithms (VQAs), in general, are prone to encountering barren plateaus~\cite{McClean2018}, we also explore how we can mitigate this issue for our algorithms by making use of the ideas presented in~\cite{Cerezo2021a}.

    Our approach is distinct from recent work on quantum algorithms for estimating entanglement.
    For example, VQAs have been used to address this problem by estimating the Hilbert--Schmidt distance~\cite{Consiglio2022}, by creating a zero-sum game using parameterized unitary circuits~\cite{Yin2022}, by employing symmetric extendibility tests~\cite{LRW21}, by estimating logarithmic negativity~\cite{arxiv.2012.14311}, and using the positive map criterion~\cite{arxiv.2012.14311}. VQAs have also been used to estimate the geometric measure of entanglement of multiqubit pure states~\cite{PhysRevApplied.18.024048}.
    The work of~\cite{AM22} is the closest related to ours, but the test used there requires two copies of the state of interest and controlled swap operations for each run of the algorithm, while our VQSA does not require either.
    In contrast, we introduce a paradigm for VQAs involving parameterized mid-circuit measurements, which is the core of our method for estimating entanglement, and we suspect that this approach will be helpful in future work for a wide variety of VQAs. Furthermore, as we show in Theorems~\ref{theorem:qip-eb_msf} and~\ref{theorem:global_cost_func}, the acceptance probabilities of our algorithms, in the ideal case, are directly related to a bonafide entanglement measure, the fidelity of separability.

\section{Results} 

\subsection{Quantum Interactive Proof for Fidelity of Separability}

\label{sec:qip-FoS}
    
    We first introduce our test for the separability of mixed states. Recall that a bipartite state is separable or unentangled if it can be written in the form given in~\eqref{eqn:sep-state-informal}, where  $\left\vert\mathcal{X}\right\vert\leq\text{rank}(\sigma_{AB})^{2}$~\cite{W89,watrous_2018}. 

    \begin{figure}
        \includegraphics[width=\columnwidth]{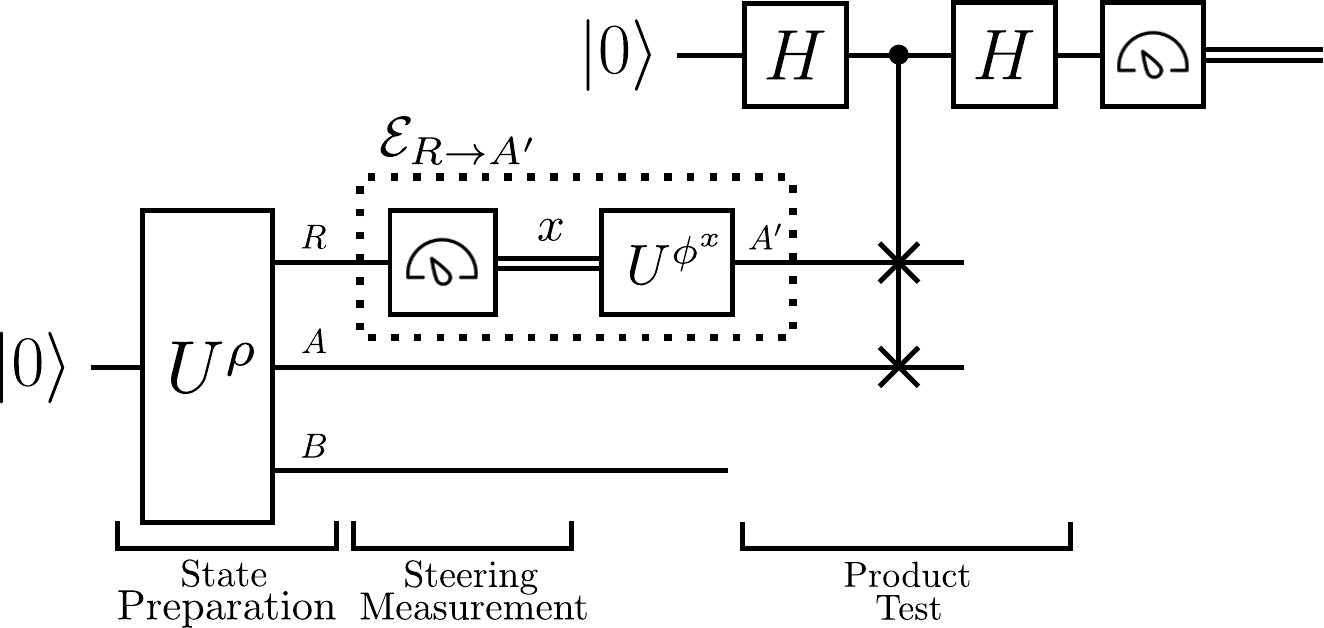}
        \caption{Test for separability of mixed states. 
        The verifier uses a unitary circuit $U^\rho$ to produce the state $\psi_{RAB}$, which is a purification of $\rho_{AB}$. The prover (indicated by the dotted box) applies an entanglement-breaking channel $\mathcal{E}_{R\rightarrow A^{\prime}}$ on $R$ by measuring the rank-one POVM $\{\mu^{x}_{R}\}_{x}$ and then, depending on the outcome $x$, prepares a pure state from the set $\{\phi^{x}_{A^{\prime}}\}_{x}$. The final state is sent to the verifier, who performs a swap test. Theorem~\ref{theorem:qip-eb_msf} states that the maximum acceptance probability of this interactive proof is equal to $\frac{1}{2}(1 + F_{s}(\rho_{AB}))$, i.e., a simple function of the fidelity of separability.}
        \label{fig:Max_Sep_Fidelity_QIP_EB}
    \end{figure}
    
    Our separability test for mixed states consists of a distributed quantum computation involving a prover and a verifier.
    The computation (depicted in Figure~\ref{fig:Max_Sep_Fidelity_QIP_EB}) begins with the verifier preparing a purification $\psi_{RAB}$\ of~$\rho_{AB}$. The verifier sends the system $R$ to a quantum prover, whom, in our model, we restrict to performing entanglement-breaking channels. The prover thus performs an entanglement-breaking channel on the reference system $R$ and sends a system $A^{\prime}$ to the verifier. An entanglement-breaking channel $\mathcal{E}_{R\rightarrow A^{\prime}}$ can always be written as a measure-and-prepare channel~\cite{HSR03}, as follows:
    \begin{equation}\label{eqn:ent_breaking}
        \mathcal{E}_{R\rightarrow A^{\prime}}(\cdot)=\sum_{x\in \mathcal{X}}\operatorname{Tr}\!\left[\mu^{x}_{R}(\cdot)\right]\phi^{x}_{A^{\prime}},
    \end{equation}
    where $\{\mu^{x}_{R}\}_{x\in \mathcal{X}}$ is a rank-one positive operator-valued measure (POVM) and $\{\phi^{x}_{A^{\prime}}\}_{x\in \mathcal{X}}$ is a set of pure states. (Due to the above measure-and-prepare decomposition of an entanglement-breaking channel, we can alternatively think of the prover as being split into two provers, a first who is allowed to perform a general quantum operation, followed by the communication of classical data to a second prover, who then is allowed to perform a general operation before communicating quantum data to the verifier. However, we proceed with the single-prover terminology in what follows.) By performing the measurement portion of the entanglement-breaking channel, the prover has, in essence, steered the verifier's systems $A$ and $B$ to a certain probabilistic ensemble of pure states. After steering the verifier's system, the prover sends system $A^\prime$ to the verifier using the preparation portion of the entanglement-breaking channel. The verifier finally performs a swap test on system $A$ and $A^\prime$ and accepts if and only if the measurement outcome of the swap test is zero. The standard model in quantum computational complexity theory~\cite{watrous2009complexity, VW15} is that the prover is always trying to get the verifier to accept the computation: in this scenario, the prover steers the verifier's systems $A$ and $B$ to an ensemble that has maximum overlap with a product-state ensemble and then sends an appropriate state to pass the swap test with the highest probability possible.
    
    The maximum acceptance probability of the distributed quantum computation detailed above is equal to
    \begin{equation}
    \label{eqn:qip-eb_accept_prob}
        \max_{\mathcal{E}\in\operatorname{EB}_{R\to A^\prime}}\operatorname{Tr}\!\left[\left(\Pi_{A^{\prime}A}^{\operatorname{sym}}\otimes I_{RB}\right)\mathcal{E}_{R\rightarrow A^{\prime}}\left(\psi_{RAB}\right)\right],
    \end{equation}
    where $\Pi_{A^{\prime}A}^{\operatorname{sym}}$ is the projector onto the symmetric subspace of the $A^{\prime}$ and $A$ systems, and  $\operatorname{EB}_{R\to A^\prime}$ denotes the set of all entanglement-breaking channels with input system $R$ and output system $A^\prime$.  We state in Theorem~\ref{theorem:qip-eb_msf} below that the maximum acceptance probability in~\eqref{eqn:qip-eb_accept_prob} can be expressed as a simple function of the fidelity of separability of $\rho_{AB}$, the latter defined as~\cite{VPRK97, VP98}
    \begin{equation}
    \label{eq:max-sep-fid-def}
        F_{s}(\rho_{AB}) \coloneqq \max_{\sigma_{AB}\in \operatorname{SEP}(A:B)} F(\rho_{AB},\sigma_{AB}) ,
    \end{equation}
    where $\operatorname{SEP}\!\left(A\!:\!B\right)$ denotes the set of separable states shared between Alice and Bob and $F(\rho,\sigma) \coloneqq \left\|\sqrt{\rho}\sqrt{\sigma}\right\|_1^2$ is the fidelity of the states $\rho$ and $\sigma$~\cite{Uhlmann1976}. The fidelity of separability is also known as the maximum separable fidelity~\cite{HMW13, HMW14, GHMW15}. With this definition, we state the first key theoretical result of our paper: 
        
    \begin{theorem}
    \label{theorem:qip-eb_msf}
        For a pure state $\psi_{RAB}$, the following equality holds:
        \begin{multline}
        \label{eqn:qip-eb_msf}
            \max_{\mathcal{E}\in\operatorname{EB}_{R\to A^\prime}}\operatorname{Tr}\!\left[\left(\Pi_{A^{\prime}A}^{\operatorname{sym}}\otimes I_{RB}\right)\mathcal{E}_{R\rightarrow A^{\prime}}\left(\psi_{RAB}\right)\right]\\
            =\frac{1 + F_{s}(\rho_{AB})}{2},
        \end{multline}
        where $F_{s}(\rho_{AB})$ is the fidelity of separability of the state $\rho_{AB} = \operatorname{Tr}_{R}[\psi_{RAB}]$.
    \end{theorem}
    
    See the first part of Section~\ref{sec:methods} for a brief overview of the proof and Appendix~\ref{appendix:proof_swap-test-eb-channel} for a detailed proof. Appendices~\ref{appendix:proof_alt_streltsov} and \ref{appendix:proof_max-sep-fid-inf-norm} recall some auxiliary results that support the proof in Appendix~\ref{appendix:proof_swap-test-eb-channel}. With this theorem, we have established a separability test for mixed states.
        

\subsection{Variational Quantum Steering Algorithm for Fidelity of Separability}
    
    We want to point out two important aspects of our separability test from Section~\ref{sec:qip-FoS}. First, note that the swap test at the end of the computation essentially leads to a measure of overlap between the state of the verifier's system and the state provided by the prover. The other important point is that, in the real world, no computationally unbounded quantum prover is available to provide the ideal states required for the tests. 

    Taking both these points into consideration, we modify the computational scenario in Figure~\ref{fig:Max_Sep_Fidelity_QIP_EB} to a)  measure the necessary overlaps directly and b) make use of quantum variational techniques~\cite{Cerezo2021} (parameterized unitary circuits and classical optimization of parameters) to approximate the actions of a computationally unbounded prover. The resulting procedure also tests and quantifies the separability of a given state by estimating its fidelity of separability. This procedure is a different quantum variational technique called a variational quantum steering algorithm (VQSA). As can be seen in Figure~\ref{fig:Max_Sep_Fidelity_VQA}, quantum steering is at the core of the VQSA via the use of a parameterized mid-circuit measurement. 

    \begin{figure}
        \includegraphics[width=\columnwidth]{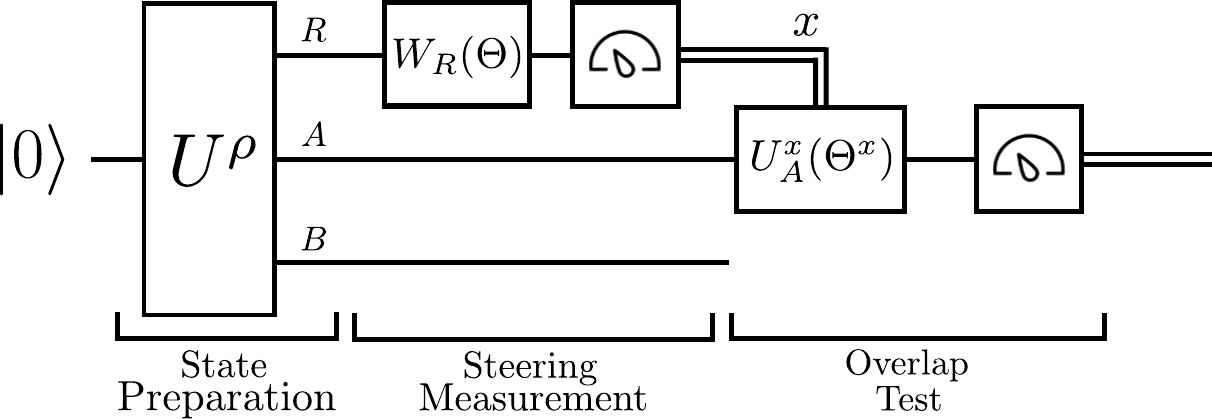}
        \caption{Quantum part of the VQSA to estimate the fidelity of separability $F_s(\rho_{AB})$. 
        The unitary circuit $U^\rho$ prepares the state $\psi_{RAB}$, which is a purification of $\rho_{AB}$. The parameterized circuit $W_R(\Theta)$ acts on $R$ to evolve $\psi_{RAB}$ to another purification of $\rho_{AB}$. The following measurement, labeled ``steering measurement,'' steers the systems $AB$ to be in a pure state~$\psi_{AB}^{x}$ if the measurement outcome~$x$ occurs. Conditioned on the outcome~$x$, the final parameterized circuit~$U^{x}_{A}(\Theta^x)$ and the subsequent measurement accepts with a maximum probability of $F_s(\rho_{AB})$.} 
        \label{fig:Max_Sep_Fidelity_VQA}
    \end{figure}

    Our VQSA is structured as follows. Let $\rho_{AB}$ denote the state for which we want to estimate the fidelity of separability, and let $\psi_{RAB}$ be a purification of it, which results from the action of the unitary operator $U^\rho$ on the all-zeros pure state $|0\rangle\!\langle 0|$. 
    Once we have $\psi_{RAB}$, we can attempt to access all possible pure-state decompositions $\{(p(x),\psi_{AB}^{x})\}_{x\in \mathcal{X}}$ of $\rho_{AB}$ by acting on system~$R$ with unitary operations. We use the first parameterized unitary $W_R(\Theta)$.  To ensure that we have a sufficient number of measurement outcomes (to cover the possible case when $|\mathcal{X}| = \text{rank}(\rho_{AB})^2$), we can prepare some ancilla qubits in the all-zeros state of a system~$R'$ and act with $W$ on $R$ and $R'$. However, without loss of generality, these extra qubits can be grouped as part of an overall reference system, relabeled as $R$. 
    
    After the action of $W_R(\Theta)$, the reference system is measured in the standard basis, and based on the outcome $x$, the post-measurement state of the system $AB$ is a pure state $\psi^x_{AB}$. We then estimate the maximum eigenvalue of the reduced state~$\psi^x_{A}$: this can be accomplished by performing a parameterized unitary $U^{x}_{A}(\Theta^x)$, based on the outcome $x$, on the reduced state~$\psi_A^x$, measuring all qubits of $A$ in the computational basis, and accepting if the all-zeros outcome occurs.
    
    Using a hybrid quantum-classical optimization loop, we can maximize the acceptance probability to estimate the value of the fidelity of separability. The quantum part of this VQSA is summarized in Figure~\ref{fig:Max_Sep_Fidelity_VQA}.

    \begin{theorem}
    \label{theorem:global_cost_func}
        If the parameterized unitary circuits involved in the quantum part of the VQSA, summarized in Figure~\ref{fig:Max_Sep_Fidelity_VQA}, can express all possible unitary operators of their respective systems, then the maximum acceptance probability of the quantum circuit is equal to $F_s(\rho_{AB})$.
    \end{theorem} 
    
    See Appendix~\ref{appendix:step_by_step} for a detailed proof.

\subsection{Benchmarking Semidefinite Programs and Examples}

   \begin{figure}
        \includegraphics[width=\columnwidth]{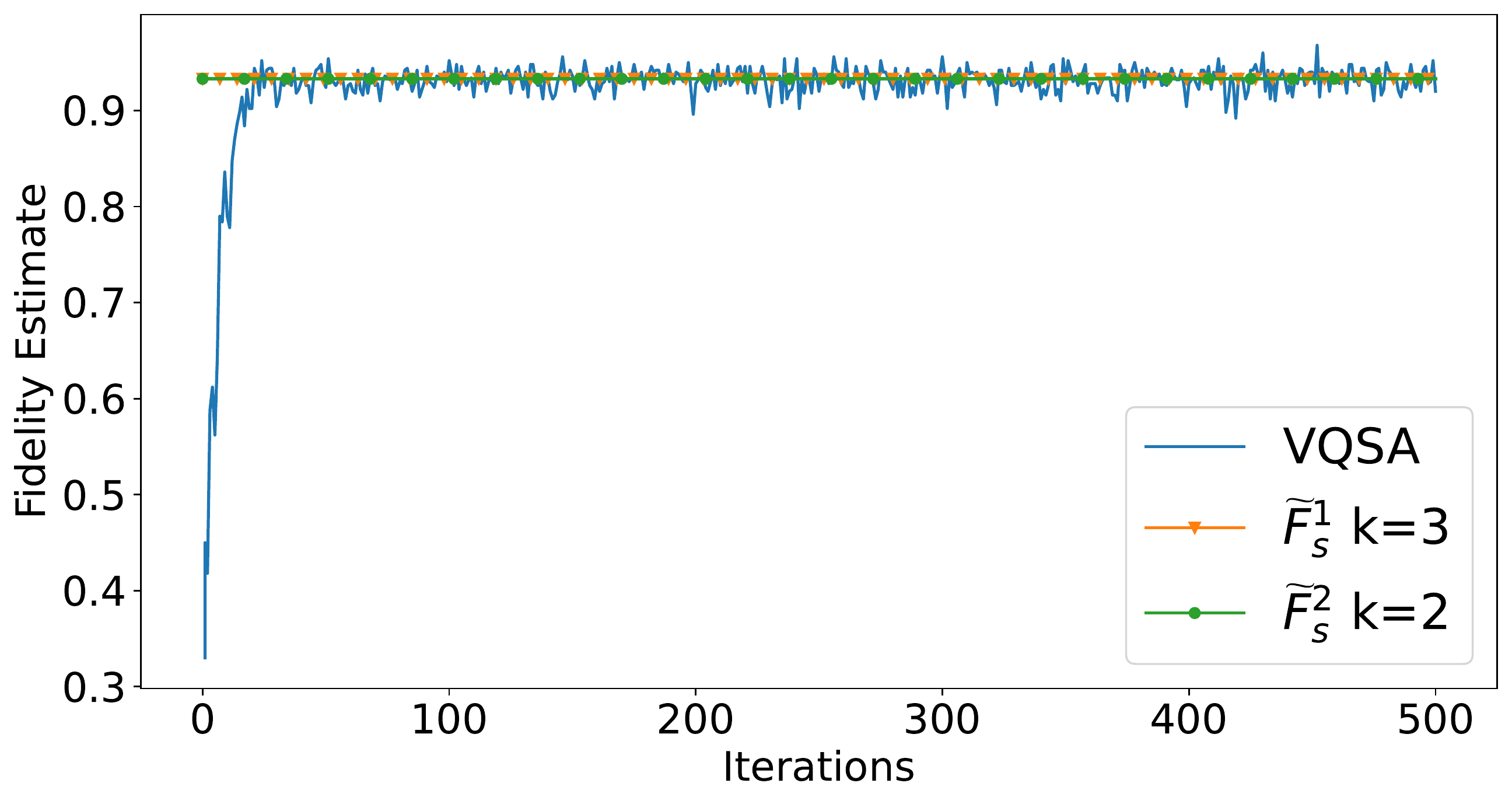}
        \caption{Fidelity of separability calculated for a ($3/4$,$1/4$) classical mixture of $|\Phi^+\rangle$ and $|\Phi^-\rangle$ using our VQSA (blue line). The algorithm converges to 0.93, which agrees with the value obtained using the benchmarks $\widetilde{F}_s^1$ and $\widetilde{F}_s^2$.}
        \label{fig:Max_Sep_Fidelity_Bell}
    \end{figure}

    Since our algorithms will be running on near-term quantum computers with limited scale and error tolerance, we develop semidefinite programs (SDPs) to benchmark the results from our VQSA because the ideal outcomes can be estimated classically for small numbers of qubits. Our benchmarks $\widetilde{F}_{s}^{1}(\rho_{AB}, k)$ and $\widetilde{F}_{s}^{2}(\rho_{AB}, k)$ are based on the positive partial transpose (PPT) and $k$-extendibility hierarchy. See details in Appendices~\ref{appendix:ppt-k-state-sdp} and \ref{appendix:proof_swap-test-ppt-k-channel-sdp}.

    We now present an example simulation of our VQSA to demonstrate that it can estimate the fidelity of separability. For our first example, we take the state of interest $\rho_{AB}$ to be a ($3/4$,$1/4$) probabilistic mixture of two maximally entangled states, $|\Phi^+\rangle = \sqrt{1/2}(|00\rangle+|11\rangle)$ and $|\Phi^-\rangle = \sqrt{1/2}(|00\rangle-|11\rangle)$, so that
    \begin{equation}
        \rho_{AB} = \frac{3}{4} |\Phi^+\rangle\!\langle \Phi^+| + \frac{1}{4} |\Phi^-\rangle\!\langle \Phi^-|.
    \end{equation}
    Systems $R$, $A$, and $B$ of the purification of $\rho_{AB}$ contain one qubit each. See Figure~\ref{fig:Max_Sep_Fidelity_Bell} for the results. We use the benchmarks and VQSA to estimate the fidelity of separability as $\approx$ 0.93. We evaluate these benchmarks for different levels of the $k$-extendibility hierarchy. See Appendix~\ref{appendix:simulations} for more examples and Appendix~\ref{appendix:software} for details about the code we developed.

   \begin{figure}
        \includegraphics[width=\columnwidth]{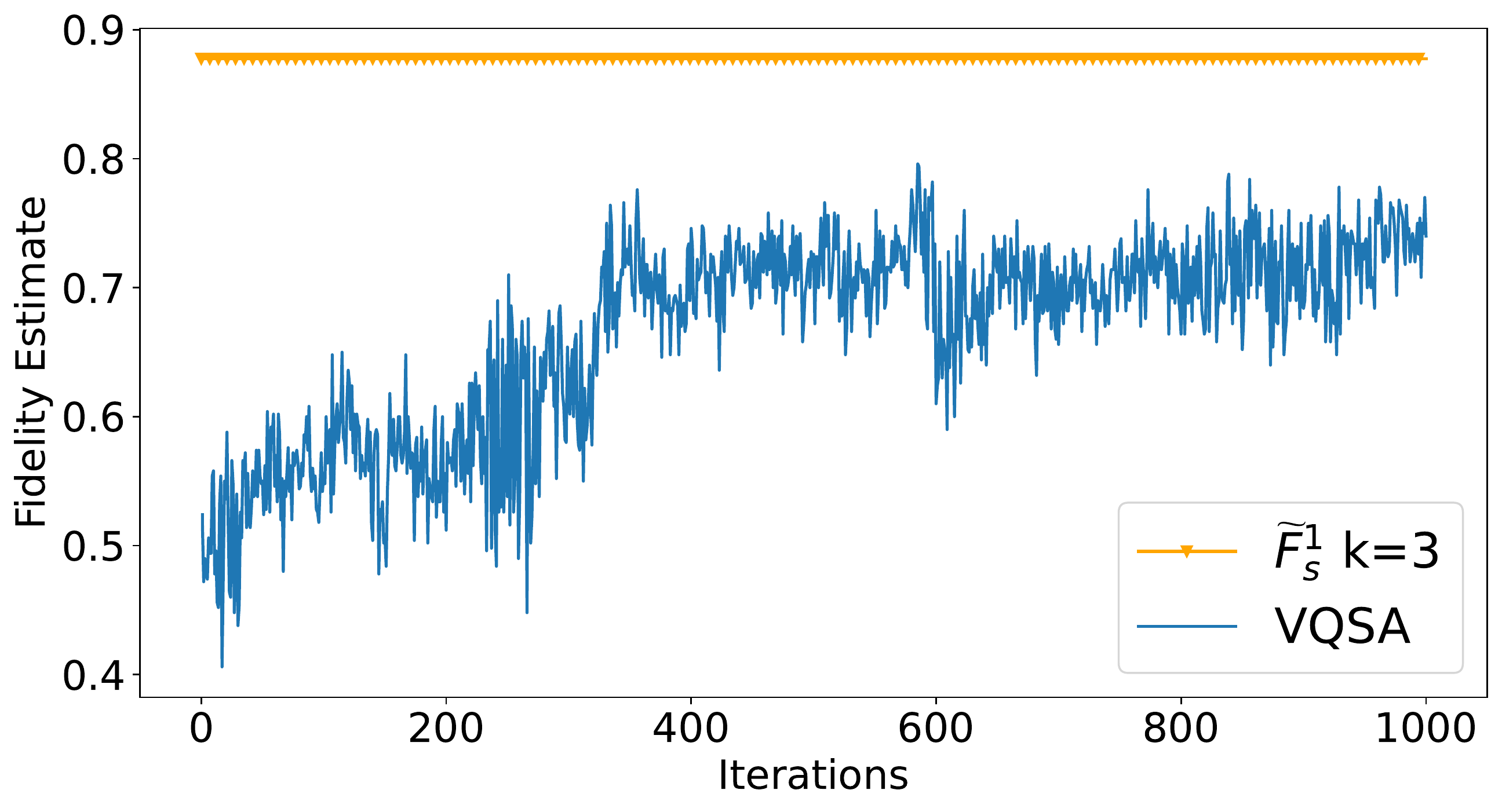}
        \caption{Fidelity of separability calculated for the state $\tilde{\rho}_{AB}$ as specified in~\eqref{eqn:rho-tilde} using our VQSA (blue line) and $\widetilde{F}_s^1$ (orange line).}
        \label{fig:Max_Sep_Fidelity_Depol}
    \end{figure}
    
    As a second example, we consider a state consisting of four qubits. Let us consider the four-qubit state $|\psi \rangle$ defined as follows:
    \begin{equation}
        \frac{1}{\sqrt{2}}\left(|0\rangle_{A_1} |0\rangle_{A_2} |0\rangle_{B_1} |0\rangle_{B_2} +|1\rangle_{A_1} |1\rangle_{A_2} |1\rangle_{B_1} |1\rangle_{B_2}\right),
    \end{equation} 
    where $A$ consists of two qubits $A_1$ and $A_2$ and $B$ consists of two qubits $B_1$ and $B_2$. We then pass $A_1$ and $A_2$ through a qubit depolarizing channel defined as $\mathcal{D}_p(\rho)\coloneqq (1-p)\rho+p\mathbb{I}/2$ where $p=0.7$. So, the final state under consideration can be written as
    \begin{equation}\label{eqn:rho-tilde}
        \tilde{\rho}_{AB}\coloneqq (\mathcal{D}_{p, A_1}\otimes\mathcal{D}_{p,A_2}\otimes\mathbb{I}_{B_1}\otimes\mathbb{I}_{B_2})\left(|\psi \rangle\!\langle\psi|\right).
    \end{equation}
    We can then use our VQSA to estimate the fidelity of separability for $\tilde{\rho}_{AB}$ and compare the result against the previous SDP benchmarks. See Figure~\ref{fig:Max_Sep_Fidelity_Depol} for the results.

\subsection{Generalization to Multipartite Fidelity of Separability}

    We also generalize our VQSA to measure the fidelity of separability of multipartite states
    in the following fashion.
   
    A multipartite state $\rho_{A_1\cdots A_M}\in \mathcal{D}(\mathcal{H}_{A_1\cdots A_M})\equiv \mathcal{D}(\mathcal{H}_{A_1}\otimes\cdots\otimes\mathcal{H}_{A_M})$ is  separable if it can be written as
        \begin{equation}
            \rho_{A_1\cdots A_M}=\sum_{x\in\mathcal{X}}p(x)\psi^{x,1}_{A_1}\otimes\cdots\otimes\psi^{x,M}_{A_M}
        \end{equation}
        where $\psi_{A_i}^{x,i}$ is a pure state for every $x\in\mathcal{X}$ and $i\in\{1,\ldots,M\}$.
    Let $M\text{-SEP}$ denote the set of all $\rho_{A_1\cdots A_M}\in \mathcal{D}(\mathcal{H}_{A_1\cdots A_M})$ such that $\rho_{A_1\cdots A_M}$ is separable. 
    
    For the multipartite case of the distributed quantum computation, the verifier prepares a purification $\psi^{\rho}_{RA_1\cdots A_M}$\ of $\rho_{A_1\cdots A_M}$. The prover applies a multipartite entanglement-breaking channel on $R$, which can be written as: 
    \begin{multline}
    \label{eqn:ent-break-mult-1}
        \mathcal{E}_{R\rightarrow A^{\prime}_1\cdots A^{\prime}_{M-1}}(\cdot)\\
        =\sum_{x\in \mathcal{X}}\operatorname{Tr}[\mu^{x}_{R}(\cdot)]\left(\phi^{x,1}_{A^{\prime}_1}\otimes\cdots\otimes\phi^{x,M-1}_{A^{\prime}_{M-1}}\right),
    \end{multline}
     where $\{\mu^{x}_{R}\}_{x}$ is a rank-one POVM and $\{\phi^{x,i}_{A^{\prime}_i}\}_{x,i}$ is a set of pure states. The prover sends systems $(A^{M-1})^{\prime}\equiv  A^{\prime}_1\cdots A^{\prime}_{M-1}$ to the verifier. (Here again, we can think of the prover as actually being split into $M$ provers, a first who performs the measurement $\{\mu^{x}_{R}\}_{x}$ and communicates the outcome $x$ to $M-1$ other provers, the $i$th of whom prepares the state $\phi^{x,i}_{A^{\prime}_i}$ and sends it to the verifier, for all $i \in \{1,\ldots, M-1\}$.) Finally, the verifier performs a collective swap test on these systems and the systems~$A_1\cdots A_M$, as depicted in Figure~\ref{fig:Max_Sep_Fidelity_QIP_EB_Multi-1}. The acceptance probability of this distributed quantum computation is given by
    \begin{equation}
        \max_{\mathcal{E}\in\operatorname{EB}_{M-1}}\operatorname{Tr}[\Pi_{(A^{M-1})^{\prime}A^{M-1}}^{\operatorname{sym}}\mathcal{E}_{R\rightarrow (A^{M-1})^{\prime}}(\psi_{RA^{M-1}})],
    \end{equation}
    where $\Pi_{(A^{M-1})^{\prime}A^{M-1}}^{\operatorname{sym}}$ is the projection onto the symmetric subspace of systems $(A^{M-1})^{\prime}$ and $A^{M-1}$ and $\operatorname{EB}_{M-1}$ denotes the set of multipartite entanglement-breaking channels defined in~\eqref{eqn:ent-break-mult-1}.
    This leads to the following theorem, which generalizes Theorem~\ref{theorem:qip-eb_msf} to the multipartite case:
    \begin{theorem}\label{theorem:qip-eb_msf_multi}
        For a pure state $\psi_{RA^{M}} \equiv \psi_{RA_1\cdots A_{M}}$, the following equality holds:%
        \begin{multline}
            \label{eq:mult-part-fid-sep-test-acc-prob-1}
            \max_{\mathcal{E}\in\operatorname{EB}_{M-1}} \operatorname{Tr}[ \Pi_{(A^{M-1})^{\prime}(A^{M-1})}^{\operatorname{sym}} \mathcal{E}_{R\rightarrow A^{\prime}_1\cdots A^{\prime}_{M-1}}(\psi_{RA^{M}})]
            \\
            =\frac{1}{2}\left(1 + F_s(\rho_{A_1\cdots A_M})\right),
        \end{multline}
        where the multipartite fidelity of separability is defined as
        \begin{multline}
            F_s(\rho_{A_1\cdots A_M})\coloneqq \\ \max_{\sigma_{A_1\cdots A_M}\in M-\operatorname{SEP}}F(\rho_{A_1\cdots A_M},\sigma_{A_1\cdots A_M}).
        \end{multline}
    \end{theorem}
    
    \begin{figure}
        \includegraphics[width=\columnwidth]{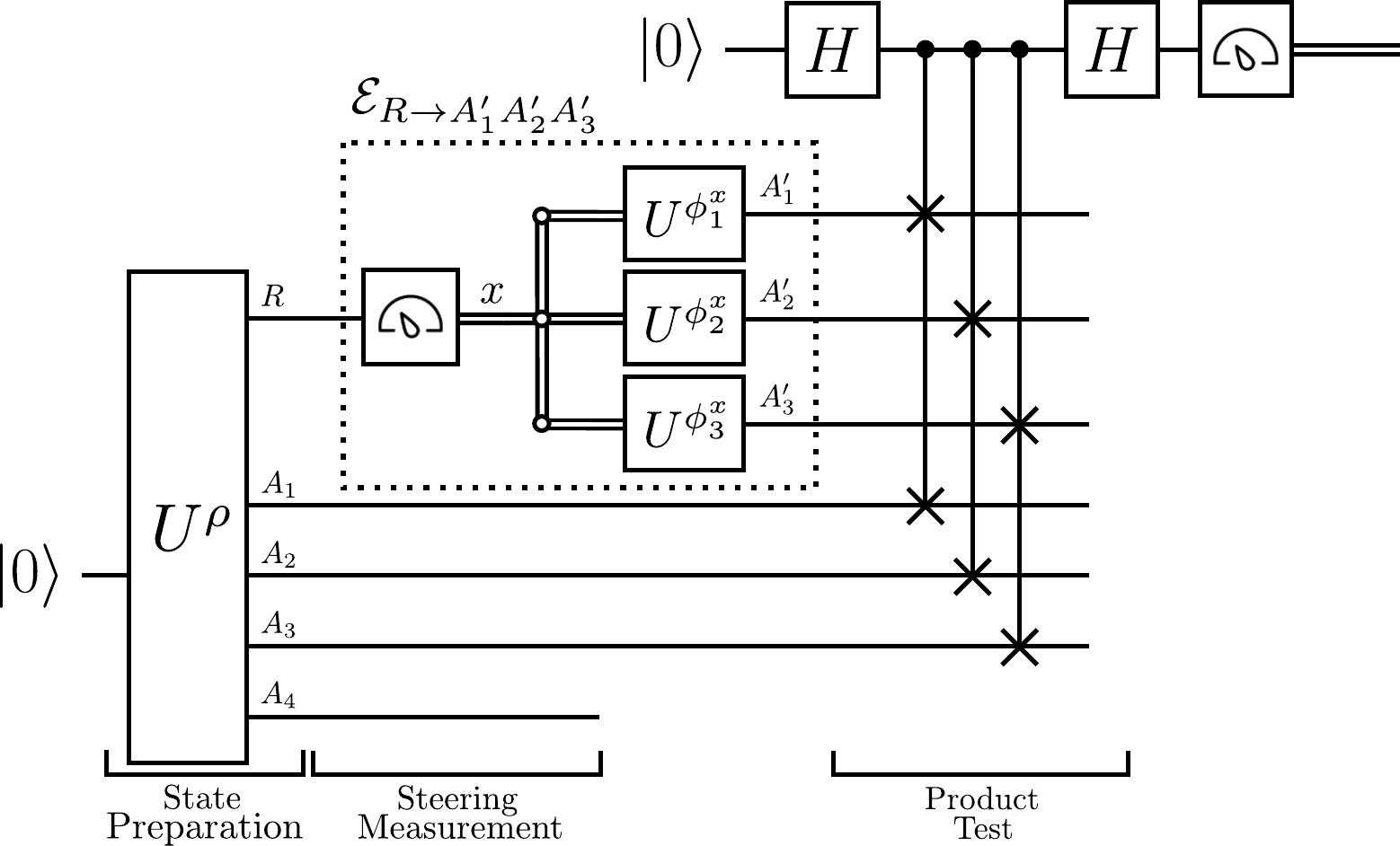}
        \caption{Test for separability of multipartite mixed states. The verifier uses the unitary circuit $U^\rho$ to prepare the state $\psi_{RA_1 A_2 A_3 A_4}$, which is a purification of $\rho_{A_1 A_2 A_3 A_4}$. The prover (indicated by the dotted box) applies an entanglement-breaking channel $\mathcal{E}_{R\rightarrow A^{\prime}_1 A^{\prime}_2 A^{\prime}_3}$ on $R$ by measuring the rank-one POVM $\{\mu^{x}_{R}\}_{x\in \mathcal{X}}$ and then, depending on the outcome $x$, prepares a state from the set $\{\phi^{x,1}_{A^{\prime}_1}\otimes\phi^{x,2}_{A^{\prime}_2}\otimes\phi^{x,3}_{A^{\prime}_3}\}_{x\in \mathcal{X}}$. The final state is sent to the verifier, who performs a collective swap test. Theorem~\ref{theorem:qip-eb_msf_multi} states that the maximum acceptance probability of this interactive proof is equal to $\frac{1}{2}(1 + F_{s}(\rho_{A_1 A_2 A_3 A_4}))$, i.e., a simple function of the fidelity of separability.}
        \label{fig:Max_Sep_Fidelity_QIP_EB_Multi-1}
    \end{figure}
     
    See Appendix~\ref{appendix:multipartite} for a proof. We can then use the generalized test of separability of mixed states to develop a VQSA for the multipartite case. See Figure~\ref{fig:Max_Sep_Fidelity_VQSA_Multi-1}. This involves replacing the collective swap test in Figure~\ref{fig:Max_Sep_Fidelity_QIP_EB_Multi-1} with an overlap measurement, similar to how we got Figure~\ref{fig:Max_Sep_Fidelity_VQA} from Figure~\ref{fig:Max_Sep_Fidelity_QIP_EB}.

    \begin{figure}
        \includegraphics[width=\columnwidth]{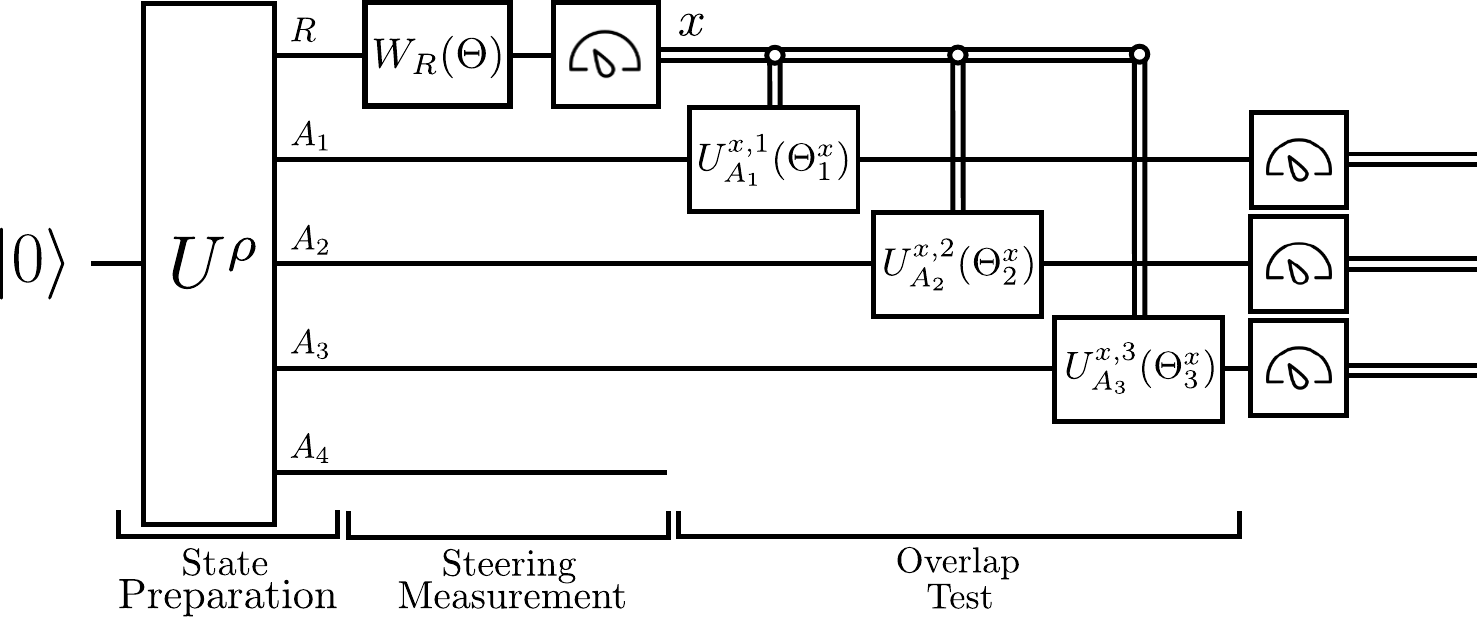}
        \caption{VQSA to estimate the multipartite fidelity of separability $F_s(\rho_{A_1A_2A_3A_4})$. 
        The unitary circuit $U^\rho$ prepares the state $\psi_{RA_1A_2A_3A_4}$, which is a purification of $\rho_{A_1A_2A_3A_4}$. The parameterized circuit $W_R(\Theta)$ acts on $R$ to evolve the state to another purification of $\rho_{A_1A_2A_3A_4}$. The following measurement, labeled ``steering measurement,'' steers the remaining systems to be in a state $\psi_{A_1A_2A_3A_4}^{x}$ if the measurement outcome $x$ occurs. Conditioned on the outcome $x$, the final parameterized circuits $U^{x,1}_{A_1}(\Theta^x_1)$, $U^{x,2}_{A_2}(\Theta^x_2)$, and $U^{x,3}_{A_3}(\Theta^x_3)$ are applied, and the subsequent measurement accepts with a maximum probability of $F_s(\rho_{A_1A_2A_3A_4})$.} 
        \label{fig:Max_Sep_Fidelity_VQSA_Multi-1}
    \end{figure}

\subsection{Quantum Computational Complexity Considerations}

    Our final result is regarding the computational complexity of estimating the fidelity of separability $F_s(\rho_{AB})$. The complexity-theoretic approach allows us to classify the separability problem based on its computational difficulty. Analyses of this form can be effectively conducted within the framework of quantum computational complexity theory~\cite{watrous2009complexity, VW15}.

     In the paradigm of complexity theory~\cite{arora_barak_2009}, a complexity class is a set of computational problems that require similar resources to solve. If a complexity class~$A$ is contained within another class $B$, then some problems in~$B$ could require more computational resources than problems in $A$. To effectively characterize the difficulty of a class of problems, we pick a problem that is representative of the class or complete for the class. A problem $h$ is considered complete for a complexity class $A$ if $h$ is contained in the class and the ability to solve problem $h$ can be extended efficiently to solve every other problem in~$A$. 

    To tackle the question posed about the computational complexity of estimating the fidelity of separability, we define \qipeb\ to be the complexity class containing problems that can be solved using a prover restricted to applying only entanglement-breaking channels, which processes a quantum message received from the verifier and sends back a quantum message to the verifier. Thus, estimating the fidelity of separability of a given state then falls within \qipeb, as seen from Figure~\ref{fig:Max_Sep_Fidelity_QIP_EB}. To fully characterize this novel complexity class, we provide a complete problem for it. We establish that, given quantum circuits to generate a channel $\mathcal{N}_{A\rightarrow B}$ and a state $\rho_{B}$, estimating the following quantity is complete for \qipeb:%
        \begin{multline}
           \max_{\substack{\{  (p(x),\psi^{x})\}  _{x},\left\{  \varphi^{x}\right\}_{x},\\
           \rho_{B} = \sum_{x}p(x)\psi_{B}^{x}
           }}   \sum_{x}p(x)F(\psi_{B}^{x},\mathcal{N}_{A\rightarrow B}(\varphi_{A}^{x}))  ,
        \end{multline}
        where $\{  (p(x),\psi^{x})\}  _{x}$ is a pure-state ensemble and $\left\{  \varphi^{x}\right\}_{x}$ is a set of pure states.
    See Appendix~\ref{appendix:qipeb} for details and an interpretation of this quantity. 

    By placing the problem of estimating the fidelity of separability in the class \qipeb, we establish results that link quantum steering and the separability problem to quantum computational complexity theory. Furthermore, we show that the complexity class \qipeb\ 
    contains QAM~\cite{marriott2004quantum} and QSZK~\cite{watrous2006zero}. It also follows, as a direct generalization of the hardness results from~\cite{HMW13, HMW14}, that the problem of estimating the fidelity of separability is hard for QSZK and NP. All of the aforementioned complexity classes are considered to be, in the worst case, out of reach of the capabilities of efficient quantum computers. See Appendix~\ref{appendix:complexity_placements} for proofs and Figure~\ref{fig:placement} for a detailed diagram. However, following the approach of~\cite{210808406}, we can try to solve some instances of problems in these classes using parameterized circuits and VQAs. 

    \begin{figure}
    \centering
        \includegraphics[width=0.8\columnwidth]{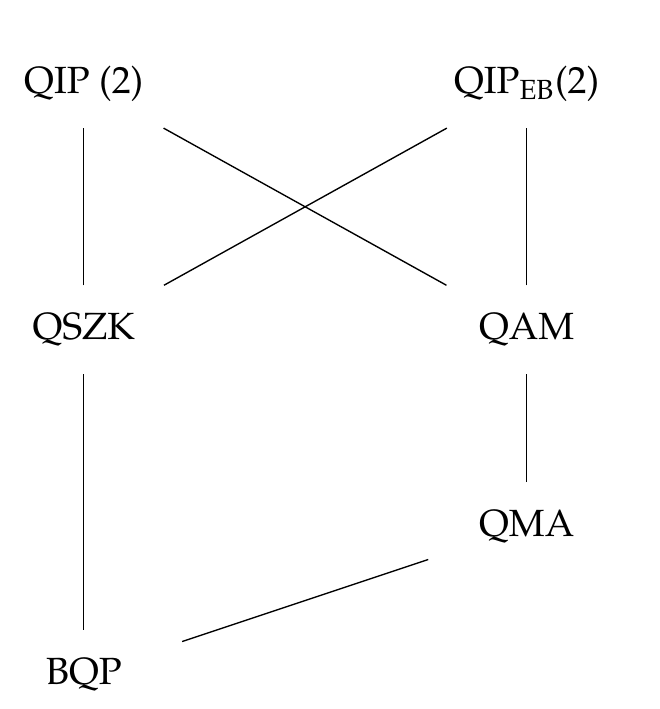}
        \caption{Placement of \qipeb\ relative to other known complexity classes. The complexity classes are organized such that if a class is connected to a class above it, the complexity class placed lower is a subset of the class above. For example, \qipeb\ is a superset of both QSZK and QAM.}
        \label{fig:placement}
    \end{figure}
    
\section{Conclusion and Discussion}  

    In this paper, we detailed a distributed quantum computation to test the separability of a quantum state that, at its core, uses quantum steering. This test demonstrated a link between quantum steering and the separability problem. The acceptance probability of this distributed quantum computation is directly related to an entanglement measure known as the fidelity of separability. 
    Using the test's structure, we also showed computational complexity-theoretic results and established a link between quantum steering, quantum algorithms, and quantum computational complexity. 
    By replacing the prover with a parameterized circuit, we modified this distributed quantum computation to develop our variational quantum steering algorithm (VQSA), a novel kind of variational quantum algorithm that uses quantum steering to address the problem of estimating the fidelity of separability. This algorithm allows for the direct estimation of the fidelity of separability without the need for state tomography and subsequent approximate tests on separability. Our algorithm is not unitary due to the mid-circuit measurement 
    on system $R$ and the consequent conditional operation applied on system $A$. This is an important distinction from most VQAs, which do not use a parameterized mid-circuit measurement. We also discuss multipartite generalizations of both our separability test and VQSA. Finally, we simulated our VQSA  using the noisy Qiskit Aer simulator \cite{Qiskit}, which showed favorable convergence trends and was compared against two classical SDP benchmarks. 

    Our VQSA has applications beyond entanglement quantification on a single quantum computer. We can also think of our VQSA as a distributed variational quantum algorithm for measuring the entanglement of a bipartite state. See~\cite{zhao2021practical,DBKC23,arxiv.2208.00450} for previous instances of distributed VQAs. Indeed, our algorithm can be executed over a quantum network in which each node has quantum and classical computers capable of performing VQAs. The initial part of the algorithm distributes $R$ to Rob, $A$ to Alice, and $B$ to Bob, who are all in distant locations. Then, Rob performs the parameterized measurement and sends the outcome over a classical channel to Alice, who performs another parameterized measurement. They can repeat this process to assess the quality of the entanglement between Alice and Bob. This interpretation is even more interesting regarding quantum networks for the multipartite case, in which the classical data gets broadcast from Rob to all the other nodes except the last one.
    
    VQSAs can tackle other problems involving quantum steering, like maximizing the pure-state decompositions of quantum states. This technique may also be helpful in estimating other entanglement measures that involve optimizing over the set of separable states. By applying the insights of \cite[Appendix~A]{Streltsov2010} and our approach here, it is clear that VQSAs will also help estimate maximal fidelities associated with other resource theories, such as the resource theory of coherence~\cite{BCP14}. More broadly, we suspect that the paradigm of parameterized mid-circuit measurements and distributed variational quantum algorithms will help address other computational problems of interest in quantum information science and physics, given recent advances in experimental implementations~\cite{Cramer2016, Egan2021, Acharya2023,graham2023midcircuit}.

    From here, we consider it an important open question in quantum computational complexity theory to place a non-trivial upper bound on the class \qipeb. As indicated in Remark~\ref{rem:de-finetti-qip}, an approach using the known quantum de Finetti theorem from \cite[Theorem~II.7']{christandl2007one} does not appear to be helpful for this task.

\section{Methods}
\label{sec:methods}

    This section briefly overviews the techniques used to prove Theorem~\ref{theorem:qip-eb_msf} (one of our main results), a brief description of our SDP benchmarks, and essential details about our simulations.
    
    To gain intuition about the separability test for mixed states, let us formulate a simple test for the separability of pure states. From~\eqref{eqn:sep-state-informal}, we can see that a pure bipartite state~$\varphi_{AB}$ is separable  if it can be written in product form, as
    \begin{equation}
     \label{eq:pure-sep-state}
     \varphi_{AB}=\psi_{A} \otimes \phi_{B},
    \end{equation}
    where $\psi_{A}$ and $ \phi_{B}$ are pure states.
    The test we developed below is important because it will reappear as part of the test for separability in the general case, along with quantum steering. Additionally, our approach slightly differs from the standard approach for testing entanglement of pure states, which employs two copies of the state in a swap test~\cite{Brennen03, HM10, GHMW15}. Instead, our approach requires only a single copy of the state.

    Our pure-state separability test consists of a distributed quantum computation involving a prover and a verifier (see Figure~\ref{fig:swaptest_pure}). 
    The computation starts with the verifier preparing the pure state $\psi_{AB}$. The prover sends the verifier the pure state $\phi_{A^\prime}$ in register $A^\prime$. (We note that the prover can send a mixed state; however, the maximum acceptance probability of the test is achieved by a pure state. Hence, without loss of generality, the prover should send a pure state.) 
    The verifier then performs the standard swap test~\cite{BBD+97, BCWW01} on $A$ and $A^\prime$ and accepts if the measurement outcome is zero. In the standard model of quantum computational complexity~\cite{watrous2009complexity, VW15}, the prover attempts to get the verifier to accept the swap test with as high a probability as possible. Thus, in this scenario, the prover selects  $\phi_{A^\prime}$ to maximize the overlap between the reduced stated $\psi_{A} \coloneqq \operatorname{Tr}_B[\psi_{AB}]$ and $\phi_{A^\prime}$. The maximum acceptance probability is then equal to
    \begin{align}
    & \max_{\phi}\operatorname{Tr}[(\Pi_{A^{\prime}A}^{\operatorname{sym}}\otimes I_B)(\phi_{A^\prime} \otimes \psi_{AB})]\notag \\
    & =
    \frac{1}{2}\left(1+\max_{\phi}\operatorname{Tr}[F_{A'A}(\phi_{A'} \otimes \psi_A)]\right) \\
    & =
    \frac{1}{2}\left(1+\max_{\phi}\operatorname{Tr}[\phi_{A}  \psi_A]\right) 
     = \frac{1}{2}\left(1+\left\|\psi_{A}\right\|_{\infty}\right),    
    \label{eq:initial-swap-test}
    \end{align}
    where $F_{A'A}$ is the unitary swap operator acting on systems $A'$ and $A$, the projector $\Pi_{A^{\prime}A}^{\operatorname{sym}} \coloneqq  \frac{1}{2} \left( I_{A'A} + F_{A'A}\right)$ projects onto the symmetric subspace of $A^{\prime}$ and $A$, and $\left\|\psi_{A}\right\|_{\infty}$ is the spectral norm of the reduced state $\psi_{A}$ (equal to its largest eigenvalue). Since $\left\|\psi_{A}\right\|_{\infty} = 1$ if and only if $\psi_{A}$ is a pure state and this occurs if and only if $\psi_{AB}$ is a product state, it follows that the maximal acceptance probability is equal to one if and only if~$\psi_{AB}$ is a product state.
    \begin{figure}
        \centering
        \includegraphics[width=0.7\columnwidth]{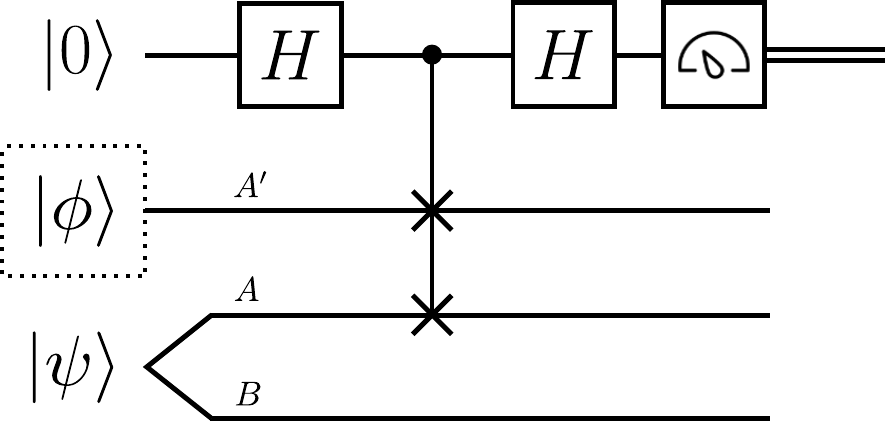}
        \caption{Pure-state separability test: The verifier has the pure state $\psi_{AB}$ of interest. The prover (indicated by the dotted box) sends the verifier a pure state $\phi_{A^\prime}$, who then performs the standard swap test on systems~$A'$ and~$A$. 
        As mentioned in~\eqref{eq:initial-swap-test}, the acceptance probability is equal to $ \frac{1}{2}(1+\left\|\psi_{A}\right\|_{\infty})$.} 
        \label{fig:swaptest_pure}
    \end{figure}

     Now we outline the proof of Theorem~\ref{theorem:qip-eb_msf}, which relies on two important facts. The first is that the fidelity of separability can be written in terms of a convex roof as follows \cite[Theorem~1]{Streltsov2010}:
        \begin{equation}
        \label{eq:convex-decomp-max-sep-fid}
            F_{s}(\rho_{AB})=  \max_{\substack{\{(p(x),\psi_{AB}^{x})\}_{x},\\\rho_{AB}=\sum_{x}p(x)\psi_{AB}^{x}}}\sum_{x}p(x)F_{s}(\psi_{AB}^{x}),
        \end{equation}
        where $\{p(x)\}_x$ is a probability distribution and each $\psi_{AB}^{x}$ is a pure state.
        See also \cite[Lemma~1]{Regula2018}.
        The second fact is that, for a pure bipartite state, $F_{s}(\psi_{AB})$ can be rewritten as \cite[Section~6.2]{Streltsov2010}
        \begin{equation}
            F_{s}(\psi_{AB})= \left\Vert \psi_{A}\right\Vert _{\infty}.
            \label{eq:max-sep-fid-inf-norm}%
        \end{equation}      
    Along with these facts, we also note that the optimization over all entanglement-breaking channels in~\eqref{eqn:qip-eb_accept_prob} is the same as optimizing over all pure-state decompositions of $\rho_{AB}$ and the rest of the proof follows. For completeness, we provide proofs of~\eqref{eq:convex-decomp-max-sep-fid} and~\eqref{eq:max-sep-fid-inf-norm} in Appendices~\ref{appendix:proof_alt_streltsov} and \ref{appendix:proof_max-sep-fid-inf-norm}, respectively. It follows from~\eqref{eqn:sep-state-informal} and~\eqref{eq:max-sep-fid-inf-norm}  that $\sum_{x}p(x)\left\Vert \psi_{A}^{x}\right\Vert _{\infty} = 1$ for a separable state, which is the maximum possible value of $F_{s}(\rho_{AB})$. Hence, the distributed quantum computation in Figure~\ref{fig:Max_Sep_Fidelity_QIP_EB} tests and quantifies the separability of a state by estimating its fidelity of separability. 
    Finally, note that the computation in Figure~\ref{fig:Max_Sep_Fidelity_QIP_EB} can be reduced to that in Figure~\ref{fig:swaptest_pure} if the purifying system $R$ is trivial, implying that the verifier only prepares a pure state on systems $A$ and $B$ in this case. 

    \textit{Benchmarking via semidefinite programs}---Here we briefly explain  the derivation of the SDP benchmarks $\widetilde{F}_{s}^{1}(\rho_{AB}, k)$ and $\widetilde{F}_{s}^{2}(\rho_{AB}, k)$.
    
    First, let us recall that the fidelity between two quantum states has an SDP formulation~\cite{Wat13}. Since there is no semidefinite constraint that directly corresponds to optimizing over the set of separable states~\cite{Fawzi2021}, we can approximate the fidelity of separability of a state by maximizing its fidelity with positive partial transpose (PPT) states~\cite{Peres1996, Horodecki1996}  and $k$-extendible states~\cite{W89a, Doherty2004}. Further noting that the PPT and $k$-extendibility constraints are semidefinite constraints, we obtain our first benchmark $\widetilde{F}_s^1(\rho_{AB}, k)$, defined in Appendix~\ref{appendix:ppt-k-state-sdp}, and which is proven there to satisfy the following bounds:
    \begin{multline}
    \label{eq:first-sdp-bounds}
        F_{s}(\rho_{AB}) \leq
        \widetilde{F}_{s}^{1}(\rho_{AB}, k)
    \\
    \leq 1- \left[\sqrt{1-F_{s}(\rho_{AB})} -2\sqrt{\frac{\left\vert B\right\vert ^{2}}{k}\left(  1-\frac{\left\vert B\right\vert ^{2}}{k}\right)  } \right]^2,
    \end{multline}
    where $\left\vert B\right\vert$ is the dimension of system $B$. By inspection of the above inequalities, observe that 
    \begin{equation}
    \lim_{k\to \infty} \widetilde{F}_{s}^{1}(\rho_{AB}, k) = F_{s}(\rho_{AB}).    
    \end{equation}
      
    The second benchmark can be obtained using~\eqref{eqn:qip-eb_accept_prob}. Just like PPT and $k$-extendible states were used to approximate separable states for the first benchmark, we use PPT channels~\cite{Rai99, Rai01} and $k$-extendible channels~\cite{PBHS13, Kaur2018, KDWW21, BBFS18} to approximate entanglement-breaking channels, leading to our second benchmark $\widetilde{F}_s^2(\rho_{AB}, k)$. We show that $\widetilde{F}_s^2(\rho_{AB}, k)$ is an SDP and approximates the fidelity of separability in the following fashion: 
        \begin{equation}
        \label{eqn:swap-test-ppt-k-channel-sdp}
         F_{s}(\rho_{AB}) \leq \widetilde{F}_s^2(\rho_{AB}, k) \leq 
            F_{s}(\rho_{AB})+\frac{4  \left\vert A\right\vert^{3} \left\vert B\right\vert}{k}.
        \end{equation}
    where $\left\vert A\right\vert$ and $\left\vert B\right\vert$ is the dimension of systems $A$ and $B$, respectively. See Appendix~\ref{appendix:proof_swap-test-ppt-k-channel-sdp} for a proof. Again, observe that
    \begin{equation}
    \lim_{k\to \infty} \widetilde{F}_{s}^{2}(\rho_{AB}, k) = F_{s}(\rho_{AB}).    
    \end{equation}

    \textit{Simulations and Reward Functions}---For our simulations, we use the Qiskit Aer simulator and Qiskit's Simultaneous Perturbation Stochastic Approximation (SPSA) optimizer to perform the classical optimization. The jitters in the fidelity values between iterations of the VQSA can be attributed to the shot noise in estimating the acceptance probability using the Qiskit Aer simulator, as well as the fact that the SPSA optimizer we have used to perform the classical optimization is itself a stochastic algorithm. We provide more examples in Appendix~\ref{appendix:simulations}.
    
    An essential issue with variational quantum techniques, such as VQAs, is the emergence of barren plateaus or vanishing gradients as the number of qubits increases~\cite{McClean2018}. However, recent results have shown that this problem can be mitigated by switching from a global reward function to a local reward function~\cite{Cerezo2021a}. In our case, a global reward function is one in which we measure all the qubits that constitute system~$A$, as done in the approach discussed in  Theorem~\ref{theorem:global_cost_func}. An example of a local reward function involves selecting a qubit in the system $A$ at random to measure in the computational basis and recording the outcome, accepting if the result equals zero. Our proposed local reward function can be used to obtain upper and lower bounds on our initial global reward function, following the approach of \cite[Appendix C]{Khatri2019quantumassisted} and discussed for completeness in Appendix~\ref{appendix:local_cost}. Local functions have also been used recently to avoid barren plateaus in VQAs when estimating the geometric measure of entanglement for pure states~\cite{zambrano2023avoiding}. We provide simulations of the local reward function in Appendix~\ref{appendix:simulations}, indicating that the local reward function can also be used to estimate the fidelity of separability of a given state.

    \section*{Acknowledgments}
        We are especially grateful to Gus Gutoski for providing the main idea of the quantum interactive proof detailed in Figure~\ref{fig:Max_Sep_Fidelity_QIP_EB}, back in September 2013. We also thank Paul Alsing, Eric Chitambar, Zo\"e Holmes, and Wilfred Salmon for insightful discussions, and Ludovico Lami, Bartosz Regula, and Alexander Streltsov for the same, as well as pointing us to~\cite{Regula2018}. AP, SR, and MMW acknowledge support from the National Science Foundation under Grant No.~1907615 and from AFRL under agreement no.~FA8750-23-2-0031.

        This material is based on research
sponsored by Air Force Research Laboratory under agreement number
FA8750-23-2-0031. The U.S.~Government is authorized to reproduce and
distribute reprints for Governmental purposes, notwithstanding any copyright
notation thereon. The views and conclusions contained herein are those of the
authors and should not be interpreted as necessarily representing the official
policies or endorsements, either expressed or implied, of Air Force Research
Laboratory or the U.S.~Government.

\bibliographystyle{quantum}
\bibliography{Ref}

\onecolumn
\appendix

\large 

\section{Proof of Theorem~\ref{theorem:qip-eb_msf}}
\label{appendix:proof_swap-test-eb-channel}
In this appendix, we prove Theorem~\ref{theorem:qip-eb_msf}, showing that the acceptance probability of the first test of separability for mixed states is equal to $\frac{1}{2}\left(1+ F_{s}(\rho_{AB})\right)$.
\bigskip 

    \begin{proof}[Proof of Theorem~\ref{theorem:qip-eb_msf}]
    Recall that an entanglement-breaking channel can be rewritten as~\cite{HSR03}
    \begin{equation}
        \mathcal{E}_{R\rightarrow A^{\prime}}(\cdot)=\sum_{x}\operatorname{Tr}[\mu^{x}_{R}(\cdot)]\phi^{x}_{A^{\prime}},
        \label{eq:eb-proof-1}
    \end{equation}
    where $\{\mu^{x}_{R}\}_{x}$ is a rank-one POVM and $\{\phi^{x}_{A^{\prime}}\}_{x}$ is a set of pure states. Then we find, for fixed~$\mathcal{E}_{R\rightarrow A^{\prime}}$, that%
    \begin{align}
         \operatorname{Tr}[\Pi_{A^{\prime}A}^{\operatorname{sym}}\mathcal{E}_{R\rightarrow A^{\prime}}(\psi_{RAB})] 
        &  =\frac{1}{2}\operatorname{Tr}[(I_{A^{\prime}A}+F_{A^{\prime}A})\mathcal{E}_{R\rightarrow A^{\prime}}(\psi_{RAB})]\\
        &=\frac{1}{2}\left(  1+\operatorname{Tr}[F_{A^{\prime}A}\mathcal{E}_{R\rightarrow A^{\prime}}(\psi_{RAB})]\right).
    \end{align}
    So let us work with the expression $\operatorname{Tr}[F_{A^{\prime}A}\mathcal{E}_{R\rightarrow A^{\prime}}(\psi_{RAB})]$. Consider that%
    \begin{align}
          \operatorname{Tr}[F_{A^{\prime}A}\mathcal{E}_{R\rightarrow A^{\prime}}(\psi_{RAB})]
        &  =\operatorname{Tr}\!\left[  F_{A^{\prime}A}\sum_{x}\operatorname{Tr}_{R}[\mu^{x}_{R}\psi_{RAB}]\otimes\phi^{x}_{A^{\prime}}\right]  \\
        &  =\operatorname{Tr}\!\left[  F_{A^{\prime}A}\sum_{x}p(x)\psi_{AB}^{x}\otimes\phi^{x}_{A^{\prime}}\right]  \\
        &  =\operatorname{Tr}\!\left[  F_{A^{\prime}A}\sum_{x}p(x)\psi_{A}^{x}\otimes\phi^{x}_{A^{\prime}}\right]  \\
        &  =\sum_{x}p(x)\langle\phi^{x}|_{A}\psi_{A}^{x}|\phi^{x}\rangle_{A},
    \end{align}
    where
    \begin{align}
        p(x) &  \coloneqq \operatorname{Tr}[\mu^{x}_{R}\psi_{RAB}], \label{eq:eb-proof-2} \\
        \psi_{AB}^{x} &  \coloneqq \frac{1}{p(x)}\operatorname{Tr}_{R}[\mu^{x}_{R}\psi_{RAB}].
        \label{eq:eb-proof-3}
    \end{align}
    Thus, the acceptance probability for a fixed entanglement-breaking channel is given by
    \begin{equation}
        \operatorname{Tr}[\Pi_{A^{\prime}A}^{\operatorname{sym}}\mathcal{E}_{R\rightarrow A^{\prime}}(\psi_{RAB})]=\frac{1}{2}\left(  1+\sum_{x}p(x)\langle\phi^{x}|_{A}\psi_{A}^{x}|\phi^{x}\rangle_{A}\right)  .
    \end{equation}
    After optimizing over every element of $\operatorname{EB}_{R\to A^\prime}$, which denotes the set of all entanglement-breaking channels with input system $R$ and output system $A^\prime$, and realizing that optimizing over measurements in $\mathcal{E}_{R\rightarrow A^{\prime}}$ induces a pure-state decomposition of $\rho_{AB}$ and optimizing over preparation channels in $\mathcal{E}_{R\rightarrow A^{\prime}}$ gives the spectral norm of $\psi _{A}^{x}$, we arrive at the claimed formula for the acceptance probability, when combined with the development in Appendices~\ref{appendix:proof_alt_streltsov} and \ref{appendix:proof_max-sep-fid-inf-norm}:%
    \begin{equation}\label{eqn:accept_prob_supp}
            \max_{\mathcal{E}\in\operatorname{EB}_{R\to A^\prime}}\operatorname{Tr}[(\Pi_{A^{\prime}A}^{\operatorname{sym}}\otimes I_{RB})\mathcal{E}_{R\rightarrow A^{\prime}}(\psi_{RAB})]=\frac{1 + F_{s}(\rho_{AB})}{2}.
        \end{equation}
    This concludes the proof.
    \end{proof}

\section{Alternative 
Proof of Equation~\eqref{eq:convex-decomp-max-sep-fid}}
\label{appendix:proof_alt_streltsov}

    This appendix provides an alternative proof for Theorem 1 in~\cite{Streltsov2010}. This proof relies on Uhlmann's theorem~\cite{Uhlmann1976}, the triangle inequality, and the Cauchy--Schwarz inequality. See also \cite[Lemma~1]{Regula2018}.
    
    \begin{theorem}[\cite{Streltsov2010}]\label{theorem:Streltsov}
        For a state $\rho_{AB}$, the following formula holds 
        \begin{equation}
            F_{s}(\rho_{AB})= \max_{\left\{(p(x),\psi_{AB}^{x})\right\}  _{x}}\left\{\sum_{x}p(x)F_{s}(\psi_{AB}^{x}):\rho_{AB}=\sum_{x}p(x)\psi_{AB}^{x}\right\},
        \end{equation}
        where the pure-state ensemble $\{(p(x),\psi_{AB}^{x})\}  _{x}$ satisfies $\sum_{x}p(x)\psi_{AB}^{x}=\rho_{AB}$, all $\psi_{AB}^{x}$ are pure, and
        \begin{equation}
        F_{s}(\psi_{AB}) = \max_{|\phi\rangle_{A} , |\varphi\rangle_{B}} |\langle\psi|_{AB} |\phi\rangle_{A}\otimes|\varphi\rangle_{B}|^2. 
        \end{equation}    
    \end{theorem}
    
    \begin{proof}
    Since the definition in~\eqref{eq:max-sep-fid-def} requires an optimization over all separable states, we take $\left\vert \mathcal{X}\right\vert =\left(\left\vert A\right\vert \left\vert B\right\vert\right)^{2}$. The separable state in~\eqref{eqn:sep-state-informal} is purified by%
    \begin{equation}
    |\psi^{\sigma}\rangle_{RAB}=\sum_{x\in\mathcal{X}}\sqrt{p(x)}|x\rangle_{R}|\psi^{x}\rangle_{A}|\phi^{x}\rangle_{B}.\label{eq:sep-purify}%
    \end{equation}
    Now consider a generic purification $|\psi^{\rho}\rangle_{R^{\prime}AB}$\ of $\rho_{AB}$. Recall that the dimension of the purifying system $R^{\prime}$ satisfies rank$(\rho_{AB})\leq\left\vert R^{\prime}\right\vert $ and so we can simply set $\left\vert R^{\prime}\right\vert =\left\vert A\right\vert\left\vert B\right\vert $. Taking $R^{\prime\prime}$ to be a system of dimension $\left\vert A\right\vert \left\vert B\right\vert $, we then have that%
    \begin{equation}
        |\psi^{\rho}\rangle_{R^{\prime}AB}|0\rangle_{R^{\prime\prime}}%
    \end{equation}
    purifies $\rho_{AB}$. Applying Uhlmann's theorem~\cite{Uhlmann1976}, the maximum separable root fidelity can be written as%
    \begin{multline}
        \max_{\sigma_{AB}\in\operatorname{SEP}(A:B)}\sqrt{F}(\rho_{AB},\sigma_{AB})
        = \\
        \max_{\substack{U,\\\left\{\left(  p(x),\psi^{x}{}_{A},\phi_{B}^{x}\right)  \right\}_{x}}} \left\vert \left(  \sum_{x^{\prime}} \sqrt{p(x^{\prime})}\langle x^{\prime}|_{R} \langle \psi^{x^{\prime}}|_{A} \langle \phi^{x^{\prime}} |_{B}\right)  \left(  U_{R^{\prime}R^{\prime\prime}\rightarrow R}\otimes I_{AB}\right) |\psi^{\rho}\rangle_{R^{\prime}AB}|0\rangle_{R^{\prime\prime}}\right\vert ,
    \end{multline}
    where the maximization is over every unitary $U_{R^{\prime}R^{\prime\prime}\rightarrow R}$. Expanding $U_{R^{\prime}R^{\prime\prime}\rightarrow R}|\psi^{\rho}\rangle_{R^{\prime}AB}|0\rangle_{R^{\prime\prime}}$ in terms of the standard basis $|x\rangle$ as%
    \begin{equation}
        U_{R^{\prime}R^{\prime\prime}\rightarrow R}|\psi^{\rho}\rangle_{R^{\prime}AB}|0\rangle_{R^{\prime\prime}}=\sum_{x\in\mathcal{X}}\sqrt{q(x)}|x\rangle_{R}|\varphi^{x}\rangle_{AB},
    \end{equation}
    we note that $U$ followed by a measurement in the standard basis induces a convex decomposition of $\rho_{AB}$ in terms of the ensemble $\{(q(x),\varphi_{AB}^{x})\}_{x}$. We can write the root fidelity as 
    \begin{align}
        \max_{\sigma_{AB}\in\operatorname{SEP}(A:B)}\sqrt{F}(\rho_{AB},\sigma_{AB})\nonumber
        &=\max_{\substack{\left\{  \left(  p(x),\psi^{x}{}_{A},\phi_{B}^{x}\right)\right\}  _{x},\\\{(q(x),\varphi_{AB}^{x})\}_{x}}} \left\vert \sum_{x,x^{\prime}}\-\sqrt{q(x)p(x^{\prime})}\langle x|x^{\prime}\rangle_{R}\langle\varphi^{x}|_{AB}|\psi^{x^{\prime}}\rangle_{A}|\phi^{x^{\prime}}\rangle_{B}\right\vert \\
        &=\max_{\substack{\left\{  \left(  p(x),\psi^{x}{}_{A},\phi_{B}^{x}\right)\right\}  _{x},\\\{(q(x),\varphi_{AB}^{x})\}_{x}}}\left\vert \sum_{x}\sqrt{q(x)p(x)}\langle\varphi^{x}|_{AB}|\psi^{x}\rangle_{A}|\phi^{x}\rangle_{B}\right\vert \\
        &=\max_{\substack{\left\{  \left(  p(x),\psi^{x}{}_{A},\phi_{B}^{x}\right)\right\}  _{x},\\\{(q(x),\varphi_{AB}^{x})\}_{x}}}\left\vert \sum_{x}\sqrt{q(x)p(x)}\langle\varphi^{x}|_{AB}|\psi^{x}\rangle_{A}|\phi^{x}\rangle_{B}\right\vert .
    \end{align}
    Next, for fixed $\left\{  \left(  p(x),\psi^{x}{}_{A},\phi_{B}^{x}\right)\right\}  _{x}$ and $\{(q(x),\varphi_{AB}^{x})\}_{x}$, we bound the objective function in the optimization above as follows:%
    \begin{align}
        \left\vert \sum_{x}\sqrt{q(x)p(x)}\langle\varphi^{x}|_{AB}|\psi^{x}\rangle_{A}|\phi^{x}\rangle_{B}\right\vert &\leq\sum_{x}\sqrt{p(x)q(x)}\left\vert \langle\varphi^{x}|_{AB}|\psi
    ^{x}\rangle_{A}|\phi^{x}\rangle_{B}\right\vert \\
    &\leq\sqrt{\sum_{x}p(x)}\sqrt{\sum_{x}q(x)\left\vert \langle\varphi
    ^{x}|_{AB}|\psi^{x}\rangle_{A}|\phi^{x}\rangle_{B}\right\vert ^{2}}\\
    &=\sqrt{\sum_{x}q(x)\left\vert \langle\varphi^{x}|_{AB}|\psi^{x}\rangle
    _{A}|\phi^{x}\rangle_{B}\right\vert ^{2}}.
    \end{align}
    The first inequality follows from the triangle inequality, and the second from an application of Cauchy--Schwarz.\ We see that equality is achieved in the second inequality by choosing
    \begin{equation}
        p(x)=\frac{q(x)\left\vert \langle\varphi^{x}|_{AB}|\psi^{x}\rangle_{A}|\phi^{x}\rangle_{B}\right\vert ^{2}}{\sum_{x}q(x)\left\vert \langle\varphi^{x}|_{AB}|\psi^{x}\rangle_{A}|\phi^{x}\rangle_{B}\right\vert ^{2}}.
    \end{equation}
    We can achieve equality in the first inequality by tuning a global phase for the state $|\psi^{x}\rangle_{A}$, which amounts to a relative phase in~\eqref{eq:sep-purify}. Putting everything together, we conclude that
    \begin{equation}
        \max_{\sigma_{AB}\in\operatorname{SEP}}F(\rho_{AB},\sigma_{AB})=\max_{\{(q(x),\varphi_{AB}^{x})\}_{x}}\sum_{x}q(x)\max_{\left(  |\psi^{x}\rangle_{A}|\phi^{x}\rangle_{B}\right)  _{x},}\left\vert \langle\varphi^{x}|_{AB}|\psi^{x}\rangle_{A}|\phi^{x}\rangle_{B}\right\vert ^{2},
    \end{equation}
    which is equivalent to the desired equality in~\eqref{eq:convex-decomp-max-sep-fid}.
\end{proof}

\section{Proof of Equation~\eqref{eq:max-sep-fid-inf-norm}}
\label{appendix:proof_max-sep-fid-inf-norm}
    
In this appendix, we show that the fidelity of separability of a bipartite state can be written in terms of the spectral norm, which was also observed in \cite[Section~6.2]{Streltsov2010}.
    \begin{proposition}
    \label{prop:max-sep-fid-inf-norm}
        For a bipartite state $\rho_{AB}$, the following equality holds
        \begin{equation}
            F_{s}(\rho_{AB})= \max_{\left\{(p(x),\psi_{AB}^{x})\right\}  _{x}}\left\{\sum_{x}p(x)\left\Vert \psi_{A}^{x}\right\Vert _{\infty}:\rho_{AB}=\sum_{x}p(x)\psi_{AB}^{x}\right\}.
        \end{equation}
    \end{proposition}
    
    \begin{proof}
    Consider that the following holds for a pure bipartite state $\psi_{AB}$:
    \begin{align}
        F_{s}(\psi_{AB}) &  =\max_{|\phi\rangle_{A},|\varphi\rangle_{B}}\left\vert\langle\psi|_{AB}|\phi\rangle_{A}\otimes|\varphi\rangle_{B}\right\vert ^{2}\\
        &=\max_{|\phi\rangle_{A},|\varphi\rangle_{B}}\left\vert \langle\phi|_{A}\otimes\langle\varphi|_{B}|\psi\rangle_{AB}\right\vert ^{2}\\
        &=\max_{|\phi\rangle_{A}}\left\Vert \langle\phi|_{A}\otimes I_{B}|\psi\rangle_{AB}\right\Vert _{2}^{2}\\
        &=\max_{|\phi\rangle_{A}}\operatorname{Tr}[(|\phi\rangle\!\langle\phi|_{A}\otimes I_{B})\psi_{AB}]\\
        &=\max_{|\phi\rangle_{A}}\operatorname{Tr}[|\phi\rangle\!\langle\phi|_{A}\psi_{A}]\\
        &=\left\Vert \psi_{A}\right\Vert _{\infty}.
    \end{align}
    The first two equalities follow from the definition and a rewriting. The third equality follows from the variational characterization of the Euclidean norm of a vector. The fourth equality follows because
    \begin{align}
         \left\Vert \langle\phi|_{A}\otimes I_{B}|\psi\rangle_{AB}\right\Vert _{2}^{2} 
        & =\left(\langle\psi|_{AB}|\phi\rangle_{A}\otimes I_{B}\right)\left(\langle\phi|_{A}\otimes I_{B}|\psi\rangle_{AB}\right)  \\
        &=\langle\psi|_{AB}|\phi\rangle\!\langle\phi|_{A}\otimes I_{B}|\psi\rangle_{AB}\\
        & =\operatorname{Tr}[(|\phi\rangle\!\langle\phi|_{A}\otimes I_{B})\psi_{AB}].
    \end{align}
    The next step follows by taking a partial trace and the final equality from the variational characterization of the spectral norm. So this implies the desired equality after applying~\eqref{eq:convex-decomp-max-sep-fid}.
    \end{proof}
    
\section{Proof of Theorem~\ref{theorem:global_cost_func}}
\label{appendix:step_by_step}
    In this appendix, we show that the acceptance probability of our VQSA is indeed equal to $F_s(\rho_{AB})$ if the parameterized unitary circuits can express all possible unitary operators of their respective systems. For this, let us track the state of the VQSA at the points indicated in Figure~\ref{fig:Max_Sep_Fidelity_costfun_proof}. 
    
    \begin{itemize}
        \item At Step $(1)$, the unitary $U^\rho$ prepares the pure state $\psi_{RAB}$. This is a specific initial purification of $\rho_{AB}$.
        
        \item At Step $(2)$, we apply the parameterized unitary circuit $W_R(\Theta)$ to $\psi_{RAB}$. Expanding $W_R(\Theta)|\psi^{\rho}\rangle_{RAB}$ in terms of the standard basis $\{|x\rangle\}_x$ leads to
            \begin{equation}
                W_R(\Theta)|\psi\rangle_{RAB}=\sum_{x\in\mathcal{X}}\sqrt{q(x)}|x\rangle_{R}|\varphi^{x}\rangle_{AB}.
            \end{equation}
            
        \item At Step $(3)$, the measurement outcome $x$ occurs with probability $q(x)$, and the state vector of registers $A$ and $B$ becomes  $|\varphi^{x}\rangle_{AB}$.
        
        \item At Step $(4)$, depending on the measurement outcome $x$, we apply the parameterized unitary circuit $U^x_A(\Theta^x)$ to register $A$. The state vector is now $U^x_A(\Theta^x)|\varphi^{x}\rangle_{AB}$.
        
        \item At Step $(5)$, we trace over $B$ and measure $A$ in the standard basis. We accept when we get the all-zeros outcome. The acceptance probability is then equal to
        \begin{equation}
            \sum_{x\in\mathcal{X}}q(x)\, \langle0|\,U^x_A(\Theta^x)\varphi^{x}_{A}\left(U^x_A\right)^\dagger|0\rangle
            =\sum_{x\in\mathcal{X}}q(x)\,\langle\phi^x|_A\varphi^{x}_{A}|\phi^x\rangle_A,
        \end{equation}
        where we have defined $|\phi^x\rangle_A \coloneqq \left(U^x_A\right)^\dagger|0\rangle$.
        
        \item Maximizing the acceptance probability corresponds to maximization over the parameters of $W_R(\Theta)$ and~$U^x_A(\Theta^x)$. 
        
        \item Maximization over the parameters of $W_R$ is a maximization over all possible pure-state decompositions of~$\rho_{AB}$. 
        
        \item Maximization over the parameters of $U^x_A(\Theta^x)$ is a maximization of $\langle\phi^x|\varphi^{x}_{A}|\phi^x\rangle$, which yields the value~$\left\Vert \varphi_{A}^{x}\right\Vert _{\infty}$.
        
        \item The maximum acceptance probability is equal to
        \begin{equation}
             \max_{\left\{(p(x),\psi_{AB}^{x})\right\}  _{x}}\left\{\sum_{x}p(x)\left\Vert \varphi_{A}^{x}\right\Vert _{\infty}:\rho_{AB}=\sum_{x}p(x)\psi_{AB}^{x}\right\} ,
        \end{equation}
        which is in turn equal to $F_s(\rho_{AB})$, by Proposition~\ref{prop:max-sep-fid-inf-norm}.
    \end{itemize} 
    
    This proves that the maximum acceptance probability equals $F_s(\rho_{AB})$ if the parameterized unitary circuits  express all possible unitary operators acting on their respective systems. However, we note that any ansatz employed for the parameterized unitary circuits has limited expressibility. As such, the maximum acceptance probability obtained via the VQSA, in principle, will also be closer to the actual value of $F_s(\rho_{AB})$ if we use a more expressive ansatz. 

    \begin{figure}
        \centering
        \includegraphics[width=0.6\columnwidth]{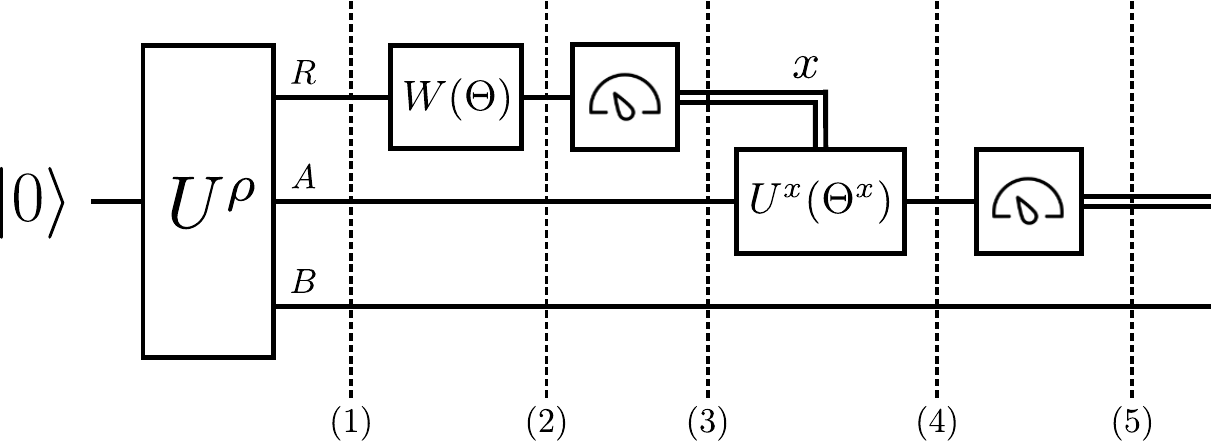}
        \caption{VQSA to estimate the fidelity of separability $F_s(\rho_{AB})$. 
        The unitary circuit $U^\rho$ produces the state $\psi_{RAB}$, which is a purification of $\rho_{AB}$. The parameterized circuit $W_R(\Theta)$ acts on $R$ to evolve $\psi_{RAB}$ to another pure-state decomposition of $\rho_{AB}$. The following measurement steers the system $AB$ to be in a pure state $\psi_{AB}^{x}$ if the measurement outcome $x$ occurs. Conditioned on the outcome $x$, the final parameterized circuit $U^{x}_{A}(\Theta^x)$ and the subsequent measurement estimates~$\left\|\psi_{A}^{x}\right\|_{\infty}$.}
        \label{fig:Max_Sep_Fidelity_costfun_proof}
    \end{figure}

\section{First Benchmarking SDP \texorpdfstring{$\widetilde{F}_s^1$}{Lg} and Proof of Equation~\eqref{eq:first-sdp-bounds}}

\label{appendix:ppt-k-state-sdp}

     This appendix details the derivation of our first benchmarking SDP $\widetilde{F}_s^1$, based on the SDP for fidelity~\cite{Wat13}. Let $\rho_{AB}$ and $\sigma_{AB}$ be bipartite states. The SDP for the root fidelity $\sqrt{F}(\rho_{AB},\sigma_{AB})$, which makes use of Uhlmann's theorem~\cite{Uhlmann1976}, is as follows:
    \begin{equation}
    \label{eqn:fidelity_sdps}
        \sqrt{F}(\rho_{AB},\sigma_{AB}) = \max_{\substack{X_{AB} \in\mathcal{L}(\mathcal{H}_{AB})}}
        \left\{\begin{array}
                [c]{c}%
                \operatorname{Re}[\operatorname{Tr}[X_{AB}]]:
                \begin{bmatrix}
                \rho_{AB} & X_{AB}\\
                X_{AB}^{\dag} & \sigma_{AB}%
                \end{bmatrix}
                \geq0 
        \end{array}\right\},
    \end{equation}
    where $\mathcal{L}(\mathcal{H}_{AB})$ is the set of all linear operators acting on the Hilbert space $\mathcal{H}_{AB}$.

    We would ideally like to include a maximization over the set of all separable states, but it is well known to be computationally challenging to optimize over this set~\cite{G03, Gharibian2010}. Furthermore, it is not generally possible to characterize the set of separable states using semi-definite constraints~\cite{Fawzi2021}. Instead, we approximate the set by constraining $\sigma_{AB}$ to have a positive partial transpose (PPT)~\cite{Peres1996, Horodecki1996} and be $k$-extendible~\cite{W89a, Doherty2004}, since all separable states satisfy these constraints. Let $\widetilde{F}_s^1(\rho_{AB})$ denote the resulting quantity, the square root of which is defined as follows:
    \begin{equation}
        \sqrt{\widetilde{F}_s^1}(\rho_{AB}, k) \coloneqq 
        \max_{\substack{X_{AB} \in\mathcal{L}(\mathcal{H}_{AB}),\\\sigma_{AB^{k}}\geq0}}
        \left\{\begin{array}
                [c]{c}
                \operatorname{Re}[\operatorname{Tr}[X_{AB}]]:\\%
                \begin{bmatrix}
                \rho_{AB} & X_{AB}\\
                X_{AB}^{\dag} & \sigma_{AB_{1}}%
                \end{bmatrix}
                \geq0,\\
                \operatorname{Tr}[\sigma_{AB^{k}}]=1,\\
                \sigma_{AB^{k}}=\mathcal{P}_{B^{k}}(\sigma_{AB^{k}}),\\ 
                T_{B_{1\cdots j}}(\sigma_{AB_{1\cdots j}})\geq 0 \quad \forall j\in \{1,\ldots, k\}
        \end{array}\right\} ,
        \label{eq:benchmark-sdp-1}
    \end{equation}
    where $B^{k}\equiv B_1 \cdots B_k$, the notation $T_R$ denotes the partial transpose map acting on system $R$, and $\mathcal{P}_{B^{k}}$ denotes the channel that performs a uniformly random permutation of systems $B_1$ through~$B_k$.
    
    We now prove the inequalities in~\eqref{eq:first-sdp-bounds}.
    Due to the containment discussed above, note that
    \begin{equation}
        F_s(\rho_{AB}) \leq \widetilde{F}_s^1(\rho_{AB}, k).
    \end{equation}
    An opposite bound on $\widetilde{F}_s^1(\rho_{AB}, k)$ in terms of $F_s(\rho_{AB})$ is as follows:
    \begin{equation}
\sqrt{1-F_{s}(\rho_{AB})}\leq
\sqrt{1-\widetilde{F}_{s}^{1}(\rho_{AB}, k)}
+2\sqrt{\frac{\left\vert B\right\vert ^{2}}{k}\left(  1-\frac{\left\vert
B\right\vert ^{2}}{k}\right)  },
\label{eq:sdp-benchmark-1-good-approx}
\end{equation}
which can be rewritten as in~\eqref{eq:first-sdp-bounds}:
\begin{equation}
    \widetilde{F}_{s}^{1}(\rho_{AB}, k)
    \leq 1- \left[\sqrt{1-F_{s}(\rho_{AB})} -2\sqrt{\frac{\left\vert B\right\vert ^{2}}{k}\left(  1-\frac{\left\vert
B\right\vert ^{2}}{k}\right)  } \right]^2.
\end{equation}
It is a consequence of \cite[Theorem~II.7']{christandl2007one}, the triangle inequality for sine distance~\cite{R06}, and  the Fuchs-van-de-Graaf inequalities~\cite{FvG99}.
    Indeed, consider that
\begin{align}
\widetilde{F}_{s}^{1}(\rho_{AB}, k)  & =\max_{\sigma_{AB}\in\text{EXT-PPT}_{k}%
}F(\rho_{AB},\sigma_{AB})
\label{eq:relaxed-benchmark-1a}\\
& \leq\max_{\sigma_{AB}\in\text{EXT}_{k}}F(\rho_{AB},\sigma_{AB}),
\label{eq:relaxed-benchmark-1}
\end{align}
where EXT-PPT$_{k}$ denotes the set of $\sigma_{AB}$ being optimized over in~\eqref{eq:benchmark-sdp-1} and
EXT$_{k}$ is the set of $k$-extendible states. Now recall that for all $\omega^k_{AB} \in \operatorname{EXT}_{k}$ (see \cite[Theorem~II.7']{christandl2007one} and also just above \cite[Theorem~II.2]{christandl2007one} for their norm convention)
\begin{equation}
\min_{\sigma_{AB}\in\operatorname{SEP}(A:B)}\frac{1}{2}\left\Vert \omega
_{AB}^{k}-\sigma_{AB}\right\Vert _{1}\leq\frac{2\left\vert B\right\vert ^{2}%
}{k},
\end{equation}
the sine distance obeys the triangle inequality~\cite{R06}:%
\begin{equation}
\sqrt{1-F(\omega,\tau)}\leq\sqrt{1-F(\omega,\xi)}+\sqrt{1-F(\xi,\tau)},
\end{equation}
and the Fuchs-van-de-Graaf inequality~\cite{FvG99}:%
\begin{equation}
1-\sqrt{F}(\omega,\tau)\leq\frac{1}{2}\left\Vert \omega-\tau\right\Vert _{1},
\end{equation}
where $\omega$, $\tau$, and $\xi$ are states. If $\frac{1}{2}\left\Vert
\omega-\tau\right\Vert _{1}\leq\varepsilon$, the latter implies that
\begin{equation}
1-\sqrt{F}(\omega,\tau)   \leq\varepsilon 
\ \Leftrightarrow \ \sqrt{1-F(\omega,\tau)}  \leq\sqrt{\varepsilon\left(
2-\varepsilon\right)  } .
\end{equation}
Letting $\sigma_{AB}^{k}$ be an optimal choice in~\eqref{eq:relaxed-benchmark-1} and $\sigma_{AB}'$ an
optimal choice for $\min_{\sigma_{AB}\in\operatorname{SEP}(A:B)}\frac{1}%
{2}\left\Vert \omega_{AB}^{k}-\sigma_{AB}\right\Vert _{1}$, this implies that%
\begin{align}
\min_{\sigma_{AB}\in\operatorname{SEP}(A:B)}\sqrt{1-F(\rho_{AB},\sigma
_{AB})}
& \leq\sqrt{1-F(\rho_{AB},\sigma_{AB}^{\prime})}\\
& \leq\sqrt{1-F(\rho_{AB},\sigma_{AB}^{k})}+\sqrt{1-F(\sigma_{AB}^{\prime
},\sigma_{AB}^{k})}\\
& \leq\sqrt{1-F(\rho_{AB},\sigma_{AB}^{k})}+2\sqrt{\frac{\left\vert
B\right\vert ^{2}}{k}\left(  1-\frac{\left\vert B\right\vert ^{2}}{k}\right)
}.
\end{align}
Rearranging and applying~\eqref{eq:relaxed-benchmark-1a}--\eqref{eq:relaxed-benchmark-1}, we arrive at the claimed inequality in~\eqref{eq:sdp-benchmark-1-good-approx}.

\section{Second Benchmarking SDP \texorpdfstring{$\widetilde{F}_s^2$}{Lg}
and Proof of Equation~\eqref{eqn:swap-test-ppt-k-channel-sdp}}
\label{appendix:proof_swap-test-ppt-k-channel-sdp}

    In this appendix, we detail the derivation of our second benchmark SDP $\widetilde{F}_s^2$, which is an SDP that approximates~\eqref{eqn:qip-eb_accept_prob} in the main text. Consider a version of the distributed quantum computation that led to~\eqref{eqn:qip-eb_accept_prob} where, instead of restricting the prover to only entanglement-breaking channels, we insist that the prover sends back $k$ systems labeled as $A_{1}\cdots A_{k}$. Then, the verifier randomly selects one of the $k$ systems and performs a swap test on the $A$ system of the state $\psi_{RAB}$. This random selection is conducted so that the prover output is effectively reduced to that of an approximate entanglement-breaking channel. Note that the resulting interactive proof is in QIP(2). More specifically, the acceptance probability of this interactive proof system is given by
    \begin{equation}\label{eqn:part_prep_channel}
       \max_{\mathcal{P}_{R\rightarrow A_{1}^{\prime}\cdots A_{k}^{\prime}}}\operatorname{Tr}[\Pi_{A^{\prime}A}^{\operatorname{sym}}\overline{\mathcal{P}}_{R\rightarrow A^{\prime}}(\psi_{RAB})],
    \end{equation}
    where
    \begin{equation}\label{eqn:whole_prep_channel}
        \overline{\mathcal{P}}_{R\rightarrow A^{\prime}}\coloneqq \frac{1}{k}\sum_{i=1}^{k}\operatorname{Tr}_{A_{1}^{k\prime}\backslash A_{i}}\circ\mathcal{P}_{R\rightarrow A_{1}^{\prime}\cdots A_{k}^{\prime}},
    \end{equation}
    and $\mathcal{P}$ is an arbitrary  channel. Observing that $\overline{\mathcal{P}}_{R\rightarrow A^{\prime}}$ is a $k$-extendible channel~\cite{PBHS13, Kaur2018, KDWW21, BBFS18}, it follows that%
    \begin{align}
    \label{eqn:part_prep_chan_k-ext}
        \max_{\mathcal{P}_{R\rightarrow A_{1}^{\prime}\cdots A_{k}^{\prime}}}\operatorname{Tr}[\Pi_{A^{\prime}A}^{\operatorname{sym}}\overline{\mathcal{P}}_{R\rightarrow A^{\prime}}(\psi_{RAB})] 
        &=\max_{\mathcal{E}_{R\rightarrow A^{\prime}}^{k}\in\text{EXT}_{k}}\operatorname{Tr}[\Pi_{A^{\prime}A}^{\operatorname{sym}}\mathcal{E}_{R\rightarrow A^{\prime}}^{k}(\psi_{RAB})],
    \end{align}
    where EXT$_{k}$ denotes the set of $k$-extendible channels. These are defined by $\mathcal{E}_{R\rightarrow A^{\prime}}^{k}(\rho_{SR})\in$ EXT$_{k}(S\!:\!A^{\prime})$ for every input state $\rho_{SR}$, where EXT$_{k}(S\!:\!A^{\prime})$ denotes the set of $k$-extendible states. Hence, the quantity in \eqref{eqn:part_prep_chan_k-ext} is an upper bound on~\eqref{eqn:qip-eb_accept_prob}, and it is given by the following SDP: 
    \begin{equation}
    \max_{\Gamma^{\mathcal{E}^{k}}_{RA^{\prime k}}\geq 0}
    \left\{\begin{array}
            [c]{c}%
            \operatorname{Tr}[\Pi_{A^{\prime}A}^{\operatorname{sym}} \operatorname{Tr}_{R}[T_R(\psi_{RAB})\Gamma^{\mathcal{E}^{k}}_{RA^{\prime}_1}]]:\\%
            \operatorname{Tr}_{A^{\prime k}}[\Gamma^{\mathcal{E}^{k}}_{RA^{\prime k}}]=I_R,\\
            \Gamma^{\mathcal{E}^{k}}_{RA^{\prime k}}=\mathcal{P}_{A^{\prime k}}(\Gamma^{\mathcal{E}^{k}}_{RA^{\prime k}})
    \end{array}\right\} ,
    \end{equation}
    where $\Gamma^{\mathcal{E}^{k}}_{RA^{\prime}}$ is the Choi operator of $\mathcal{E}^{k}$ and $\mathcal{P}_{A^{\prime k}}$ is the channel that randomly permutes the systems~$A^{\prime k}$. We can add further PPT constraints to this SDP, which is still an upper bound on~\eqref{eqn:qip-eb_accept_prob} and leads to our second benchmark $\widetilde{F}_s^2(\rho_{AB}, k)$:
    \begin{equation}
    \frac{1}{2}(1+\widetilde{F}_s^2(\rho_{AB}, k)) \coloneqq 
    \max_{\Gamma^{\mathcal{E}^{k}}_{RA^{\prime k}}\geq 0}
    \left\{\begin{array}
            [c]{c}%
            \operatorname{Tr}[\Pi_{A^{\prime}A}^{\operatorname{sym}} \operatorname{Tr}_{R}[T_R(\psi_{RAB})\Gamma^{\mathcal{E}^{k}}_{RA^{\prime}_1}]]:\\%
            \operatorname{Tr}_{A^{\prime k}}[\Gamma^{\mathcal{E}^{k}}_{RA^{\prime k}}]=I_R,\\
            \Gamma^{\mathcal{E}^{k}}_{RA^{\prime k}}=\mathcal{P}_{A^{\prime k}}(\Gamma^{\mathcal{E}^{k}}_{RA^{\prime k}}),\\ 
            T_{A^{\prime}_{1\cdots j}}(\Gamma^{\mathcal{E}^{k}_{RA^{\prime k}}}) \geq 0 \quad \forall j \in \{1,\ldots,  k\}
    \end{array}\right\} 
    \label{eq:2nd-bench-def-PPT-EXT}
    \end{equation}
    where the map $T_R$ is the partial transpose map acting on system~$R$.
    
    The following theorem indicates how $\widetilde{F}_s^2$ approximates $F_{s}(\rho_{AB})$.

    \begin{proposition}
    \label{prop:swap-test-ppt-k-channel-sdp}
    The following bound holds for a bipartite state~$\rho_{AB}$:
        \begin{equation}
            F_{s}(\rho_{AB}) \leq \widetilde{F}_s^2(\rho_{AB}, k) \leq 
            F_{s}(\rho_{AB})+\frac{4  \left\vert A\right\vert^{3} \left\vert B\right\vert}{k}.
        \end{equation}
    \end{proposition}
    \begin{proof}
    Since every entanglement-breaking channel is $k$-extendible, we trivially find that%
    \begin{align}
        \frac{1+F_{s}(\rho_{AB})}{2} &=\max_{\mathcal{E}\in\operatorname{EB}}\operatorname{Tr}[(\Pi_{A^{\prime}A}^{\operatorname{sym}}\otimes I_{RB})\mathcal{E}_{R\rightarrow A^{\prime}}(\psi_{RAB})]\\
        & \leq\max_{\mathcal{E}_{R\rightarrow A^{\prime}}^{k}\in\text{EXT-PPT}_{k}}\operatorname{Tr}[(\Pi_{A^{\prime}A}^{\operatorname{sym}}\otimes I_{RB})\mathcal{E}_{R\rightarrow A^{\prime}}^{k}(\psi_{RAB})]
        \\
        & = \frac{1+\widetilde{F}_s^2(\rho_{AB}, k)}{2},
    \end{align}
    where EXT-PPT$_{k}$ denotes the set of channels satisfying the constraints in~\eqref{eq:2nd-bench-def-PPT-EXT}. 
    Consider the following bound for a $k$-extendible state $\omega_{AB}^{k}$ \cite[Theorem~II.7'] {christandl2007one} (see also just above \cite[Theorem~II.2]{christandl2007one} for their norm convention):
    \begin{equation}
        \min_{\sigma_{AB}\in\operatorname{SEP}(A:B)}\frac{1}{2}\left\Vert\omega_{AB}^{k}-\sigma_{AB}\right\Vert _{1}\leq\frac{2\left\vert B\right\vert ^{2}}{k}.
    \end{equation}
    We can use it and the result of \cite[Lemma~7]{Wallman_2014} to conclude that
    \begin{equation}
        \min_{\mathcal{E}\in\operatorname{EB}}\frac{1}{2}\left\Vert \mathcal{E}^{k}-\mathcal{E}\right\Vert _{\diamond}\leq\frac{2\left\vert R\right\vert\left\vert A^{\prime}\right\vert ^{2}}{k}.
    \end{equation}
    Then consider that, for every fixed choice of $\mathcal{E}_{R\rightarrow A^{\prime}}^{k}\in$ EXT-PPT$_{k} $, there exists an entanglement-breaking channel $\mathcal{E}$ satisfying%
    \begin{equation}
        \frac{1}{2}\left\Vert \mathcal{E}^{k}-\mathcal{E}\right\Vert _{\diamond}\leq\frac{2\left\vert R\right\vert \left\vert A^{\prime}\right\vert ^{2}}{k}.
    \end{equation}
    Then we find that
    \begin{align}
        & \operatorname{Tr}[(\Pi_{A^{\prime}A}^{\operatorname{sym}}\otimes I_{RB})\mathcal{E}_{R\rightarrow A^{\prime}}^{k}(\psi_{RAB})]\notag\\
        &\leq\operatorname{Tr}[(\Pi_{A^{\prime}A}^{\operatorname{sym}}\otimes I_{RB})\mathcal{E}_{R\rightarrow A^{\prime}}(\psi_{RAB})]+\frac{2\left\vert R\right\vert \left\vert A^{\prime}\right\vert ^{2}}{k}\\
        &\leq\max_{\mathcal{E}\in\operatorname{EB}}\operatorname{Tr}[(\Pi_{A^{\prime}A}^{\operatorname{sym}}\otimes I_{RB})\mathcal{E}_{R\rightarrow A^{\prime}}(\psi_{RAB})]+\frac{2\left\vert R\right\vert \left\vert A^{\prime}\right\vert ^{2}}{k}\\
        &=\frac{1+F_{s}(\rho_{AB})}{2}+\frac{2\left\vert R\right\vert \left\vert A^{\prime}\right\vert ^{2}}{k}.
    \end{align}
    Since the inequality holds for every $\mathcal{E}_{R\rightarrow A^{\prime}}^{k}\in$ EXT-PPT$_{k}$, it follows that
    \begin{equation}
        \max_{\mathcal{E}_{R\rightarrow A^{\prime}}^{k}\in\text{EXT-PPT}_{k}}\operatorname{Tr}[(\Pi_{A^{\prime}A}^{\operatorname{sym}}\otimes I_{RB})\mathcal{E}_{R\rightarrow A^{\prime}}^{k}(\psi_{RAB})]
        \leq\frac{1+F_{s}(\rho_{AB})}{2}+\frac{2\left\vert R\right\vert \left\vert A^{\prime}\right\vert ^{2}}{k}.
    \end{equation}
    This concludes the proof after recalling that $|R| \leq |A| |B|$, observing that $|A|=|A'|$, and performing some simple algebra.
    \end{proof}

    \begin{remark}
    \label{rem:de-finetti-qip}
        Although the correction term in the upper bound in Proposition~\ref{prop:swap-test-ppt-k-channel-sdp} decreases with increasing $k$, it is clear that, for it to become arbitrarily small,  $k$ needs to be larger than $|A|^3|B|$, which is exponential in the number of qubits for the state $\rho_{AB}$. Thus, this approach does not lead to an efficient method for placing the fidelity of separability estimation problem in QIP(2) or even QIP.
    \end{remark}

\section{Further Simulations and Details}
\label{appendix:simulations}

    \begin{figure}
        \centering
        \subfigure[Fidelity of separability calculated for a random product state using the local reward function of the VQSA and benchmarked by $\widetilde{F}_s^1$.]{\includegraphics[width=0.48\columnwidth]{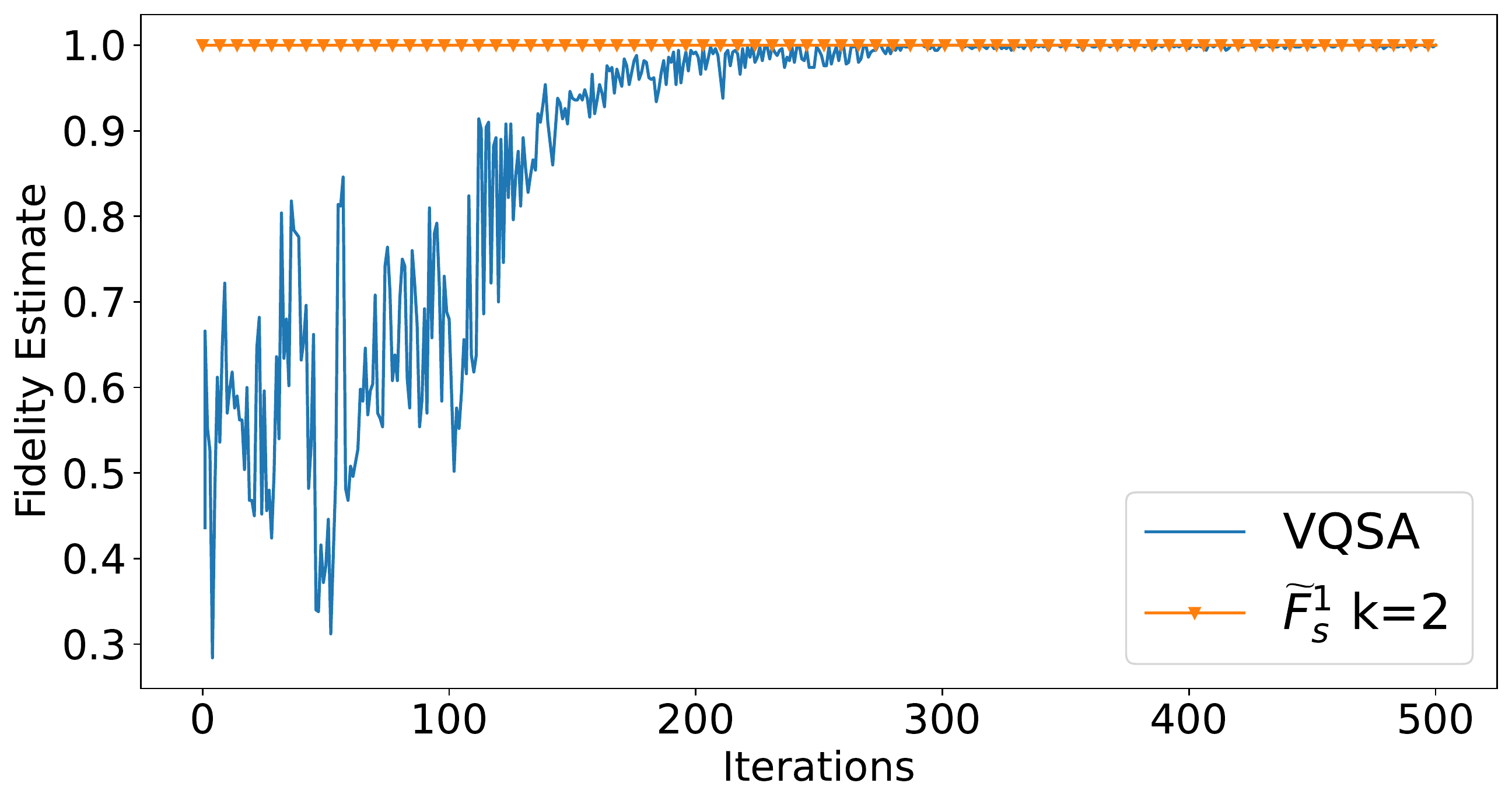}\label{fig:local_cost_function1}}
            \hspace*{\fill}
        \subfigure[Fidelity of separability calculated for a random entangled state using the local reward function of the VQSA and benchmarked by $\widetilde{F}_s^1$.]{\includegraphics[width=0.48\columnwidth]{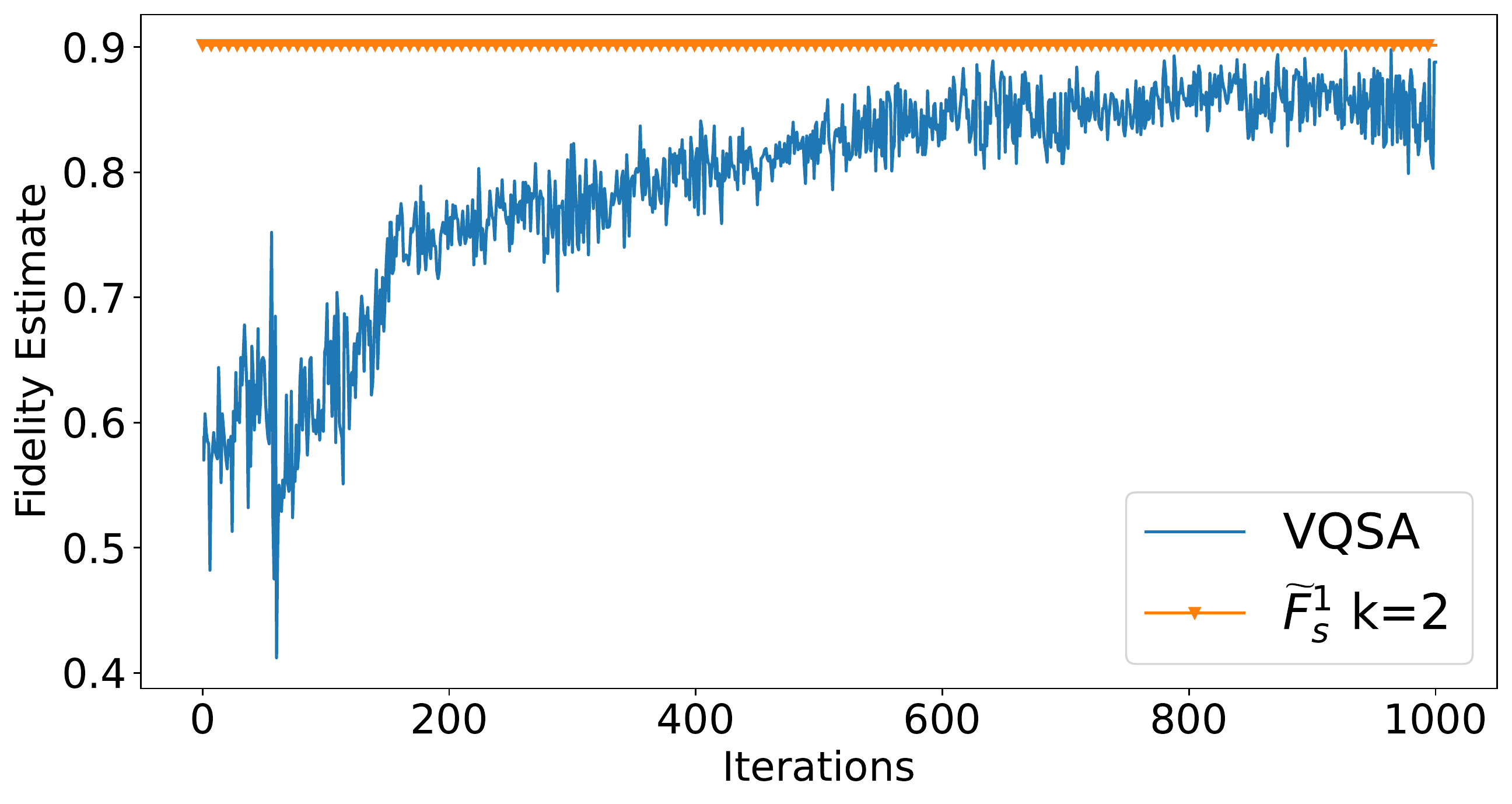}\label{fig:local_cost_function2}}
        \caption{Fidelity of separability estimated using the local reward function of the VQSA and benchmarked by~$\widetilde{F}_s^1$.} 
    \end{figure}
    
    The input states and parameterized unitaries were generated using the hardware efficient ansatz (HEA)~\cite{KMTTBCG17} for all the simulations in our work. The HEA consists of several layers, each composed of two parameters per qubit per layer, specifying rotations about the $x$- and $y$-axes. After each layer of rotations is a series of neighboring qubit CNOT gates. When using the HEA to generate the input states, we keep the rotation angles fixed, thus leading to a fixed input state. For the parameterized unitaries, the rotation angles are parameters and are optimized over.
    
    In Figure~\ref{fig:local_cost_function1}, we report simulation results after generating a random bipartite product state, with each partition containing two qubits. We remove all the CNOT gates from the HEA that generates the input state $\rho$ to guarantee a product state. We calculated the fidelity of separability using both the local reward function of the VQSA and the benchmark $\widetilde{F}^1_s$, the latter discussed in Appendix~\ref{appendix:ppt-k-state-sdp}. 
    
    In Figure~\ref{fig:local_cost_function2}, we do the same for a random bipartite state with the partitions $A$ and $B$ containing two qubits and one qubit, respectively, and three qubits in the reference system.  
    
    We generated all parameterized unitary circuits in the following fashion. We used the Qiskit Aer simulator and Qiskit's Simultaneous Perturbation Stochastic Approximation (SPSA) optimizer to perform the classical optimization. All other details can be found in Table~\ref{tab:details_about_fig}. The local reward function of the VQSA requires more classical processing (like picking a qubit at random to measure) and seems to require more iterations to reach the correct value. However, these downsides are outweighed by the fact that it is less susceptible to the emergence of barren plateaus. More details about the local cost function can be found in Appendix~\ref{appendix:local_cost}.
    
    \renewcommand{\arraystretch}{1.25}
    \begin{table*}
    \centering
    \begin{tabular}{|c|c|c|c|}
    \hline
    Figure & No. of Qubits & State $\rho_{AB}$ & Layer Count \\
    \hline\hline
    \multirow{1}{*}[-0.8em]{\ref{fig:Max_Sep_Fidelity_Bell}} & \multirow{1}{*}[-0.8em]{$R=2$, $A=1$, $B=1$} & \multirow{1}{*}[-0.8em]{$(3/4)|\Phi^+\rangle\!\langle\Phi^+|+(1/4)|\Phi^-\rangle\!\langle\Phi^-|$} & $W_R$ no.~of layers = 2\\
    & & & $U^x_A$ no.~of layers = 2\\
    \hline
    \multirow{1}{*}[-0.8em]{\ref{fig:Max_Sep_Fidelity_Depol}} & \multirow{1}{*}[-0.8em]{$R=4$, $A=2$, $B=2$} & $(\mathcal{D}_{p, A_1}\otimes\mathcal{D}_{p,A_2}\otimes\mathbb{I}_{B})\left(|\psi \rangle\!\langle\psi|\right)$ & $W_R$ no.~of layers = 4\\
    & & $|\psi \rangle=\frac{1}{\sqrt{2}}\left(|0\rangle_{A_1} |0\rangle_{A_2} |00\rangle_{B} +|1\rangle_{A_1} |1\rangle_{A_2} |11\rangle_{B}\right)$ & $U^x_A$ no.~of layers = 4\\
    \hline
    \multirow{1}{*}[-0.8em]{\ref{fig:local_cost_function1}} & \multirow{1}{*}[-0.8em]{$R=3$, $A=2$, $B=2$} & \multirow{1}{*}[-0.8em]{Random product state using HEA~\cite{KMTTBCG17}} & $W_R$ no.~of layers = 4\\
    & & & $U^x_A$ no.~of layers = 4\\
    \hline
    \multirow{1}{*}[-0.8em]{\ref{fig:local_cost_function2}} & \multirow{1}{*}[-0.8em]{$R=3$, $A=2$, $B=2$} & \multirow{1}{*}[-0.8em]{Random entangled state using HEA~\cite{KMTTBCG17}} &$W_R$ no.~of layers = 4\\
    & & & $U^x_A$ no.~of layers = 4\\
    \hline
    
    \end{tabular}
    \caption{Details of all VQSA simulations.}
    \label{tab:details_about_fig}
\end{table*}
\renewcommand{\arraystretch}{1.0}

\section{Software}
\label{appendix:software}
All of our Python source files are available with the arXiv posting of this paper. We performed all simulations using the noisy Qiskit Aer simulator. The Picos Python package~\cite{sagnol2012picos} was used to invoke the CVXOPT solver~\cite{vandenberghe2010cvxopt} for solving the SDPs, and the toqito Python package~\cite{russo2021toqito} was used for carrying out specific operations on the matrices representing quantum systems.      

\section{Multipartite Scenarios}
\label{appendix:multipartite}

    This appendix discusses a multipartite generalization of the separability tests for mixed states.

    \begin{definition}
        A state $\rho_{A_1\cdots A_M}\in \mathcal{D}(\mathcal{H}_{A_1\cdots A_M})=\mathcal{D}(\mathcal{H}_{A_1}\otimes\cdots\otimes\mathcal{H}_{A_M})$ is  separable if it can be written as
        \begin{equation}
            \rho_{A_1\cdots A_M}=\sum_{x\in\mathcal{X}}p(x)\psi^{x,1}_{A_1}\otimes\cdots\otimes\psi^{x,M}_{A_M} ,
        \end{equation}
        where $\{p(x)\}_{x\in\mathcal{X}}$ is a probability distribution and $\psi_{A_i}^{x,i}$ is a pure state for all $x\in\mathcal{X}$ and $i\in\{1,\ldots,M\}$.
    \end{definition}
    
    Let $M\text{-SEP}$ denote the set of all $\rho_{A_1\cdots A_M}\in \mathcal{D}(\mathcal{H}_{A_1\cdots A_M})$ such that $\rho_{A_1\cdots A_M}$ is separable. The following theorem is important for the rest of this analysis.
    
    \begin{theorem}[\cite{Streltsov2010}]\label{theorem:Multi-Streltsov}
        The following formula holds 
        \begin{equation}
        \label{eq:convex-decomp-multi-max-sep-fid}
            \max_{\sigma_{A_1\cdots A_M}\in M\text{-}\operatorname{SEP}}F(\rho_{A_1\cdots A_M},\sigma_{A_1\cdots A_M})
            =\max_{\{(q(x),\varphi_{A_1\cdots A_M}^{x})\}_{x}}\sum_{x}q(x) F_s(\varphi^{x}_{A_1\cdots A_M}),
        \end{equation}
        where the optimization is over every pure-state decomposition $\{(q(x),\varphi_{A_1\cdots A_M}^{x})\}_{x}$   of $\rho_{A_1\cdots A_M}$ (similar to those in Theorem~\ref{theorem:Streltsov}) and 
        \begin{equation}
        \label{eqn:multi_fid_sep}
            F_{s}(\varphi_{A_1\cdots A_M}^{x}) =
            \max_{\left\{|\phi^{x,i}\rangle_{A_i}\right\}_{i=1}^{M}}\left\vert\langle\varphi^{x}|_{A_1\cdots A_M}|\phi^{x,1}\rangle_{A_1}\otimes\cdots\otimes |\phi^{x,M}\rangle_{A_M}\right\vert^2.
        \end{equation}
    \end{theorem}
    
    For the multipartite case of the distributed quantum computation, the verifier prepares a purification $\psi^{\rho}_{R A_1\cdots A_M}$\ of $\rho_{A_1\cdots A_M}$. The prover applies a multipartite entanglement-breaking channel on system~$R$, which can be written as
    \begin{equation}
    \label{eqn:ent-break-mult}
        \mathcal{E}_{R\rightarrow A^{\prime}_1\cdots A^{\prime}_{M-1}}(\cdot)=\sum_{x}\operatorname{Tr}[\mu^{x}_{R}(\cdot)]\left(\phi^{x,1}_{A^{\prime}_1}\otimes\cdots\otimes\phi^{x,M-1}_{A^{\prime}_{M-1}}\right),
    \end{equation}
     where $\{\mu^{x}_{R}\}_{x}$ is a rank-one POVM and $\{\phi^{x,i}_{A^{\prime}_i}\}_{x,i}$ is a set of pure states. The prover sends systems $(A^{M-1})^{\prime}\equiv  A^{\prime}_1\cdots A^{\prime}_{M-1}$ to the verifier. Now, the verifier performs a collective swap test of these systems with $A_1\cdots A_M$, as depicted in Figure~\ref{fig:Max_Sep_Fidelity_QIP_EB_Multi}. The acceptance probability of this distributed quantum computation is given by
    \begin{equation}
        \max_{\mathcal{E}\in\operatorname{EB}_{M-1}}\operatorname{Tr}[\Pi_{(A^{M-1})^{\prime}A^{M-1}}^{\operatorname{sym}}\mathcal{E}_{R\rightarrow (A^{M-1})^{\prime}}(\psi_{RA^{M-1}})],
    \end{equation}
    where $\mathcal{E}\in\operatorname{EB}_{M-1}$ denotes the set of entanglement-breaking channels defined in~\eqref{eqn:ent-break-mult}.
    This leads to the following theorem:
    \begin{theorem}\label{theorem:qip-eb_msf_multi-app}
        For a pure state $\psi_{RA^{M}} \equiv \psi_{RA_1\cdots A_{M}}$, the following equality holds:%
        \begin{equation}
            \label{eq:mult-part-fid-sep-test-acc-prob}
            \max_{\mathcal{E}\in\operatorname{EB}_{M-1}} \operatorname{Tr}[ \Pi_{(A^{M-1})^{\prime}(A^{M-1})}^{\operatorname{sym}} \mathcal{E}_{R\rightarrow A^{\prime}_1\cdots A^{\prime}_{M-1}}(\psi_{RA^{M}})]
            =\frac{1}{2}\left(1 +\max_{\sigma_{A_1\cdots A_M}\in M\text{-}\operatorname{SEP}}F(\rho_{A_1\cdots A_M},\sigma_{A_1\cdots A_M})\right).
        \end{equation}
    \end{theorem}
    
    \begin{figure}
        \centering
        \includegraphics[width=0.7\columnwidth]{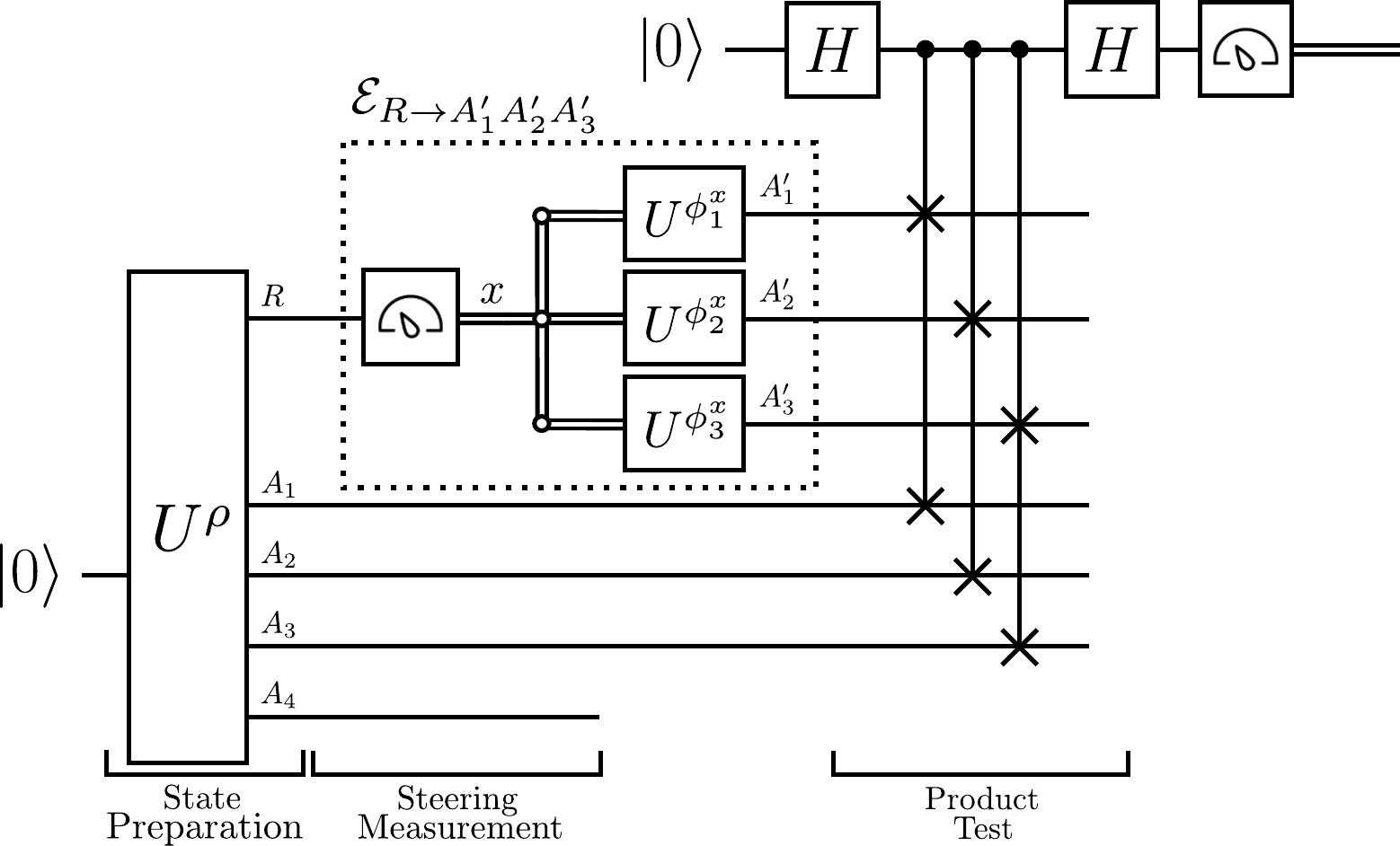}
        \caption{Test for separability of multipartite mixed states. The verifier uses a unitary circuit $U^\rho$ to produce the state $\psi_{RA_1 A_2 A_3 A_4}$, which is a purification of $\rho_{A_1 A_2 A_3 A_4}$. The prover (indicated by the dotted box) applies an entanglement-breaking channel $\mathcal{E}_{R\rightarrow A^{\prime}_1 A^{\prime}_2 A^{\prime}_3}$ on $R$ by measuring the rank-one POVM $\{\mu^{x}_{R}\}_{x}$ and then, depending on the outcome $x$, prepares a state from the set $\{\phi^{x,1}_{A^{\prime}_1}\otimes\phi^{x,2}_{A^{\prime}_2}\otimes\phi^{x,3}_{A^{\prime}_3}\}_{x}$. The final state is sent to the verifier, who performs a collective swap test. Theorem~\ref{theorem:qip-eb_msf_multi-app} states that the maximum acceptance probability of this interactive proof is equal to $\frac{1}{2}(1 + F_{s}(\rho_{A_1 A_2 A_3 A_4}))$, i.e., a simple function of the multipartite fidelity of separability.}
        \label{fig:Max_Sep_Fidelity_QIP_EB_Multi}
    \end{figure}
    
    \begin{proof}
        The circuit diagram is given in Figure~\ref{fig:Max_Sep_Fidelity_QIP_EB_Multi}. The verifier prepares a purification $\psi^{\rho}_{RA_1\cdots A_M}$\ of $\rho_{A_1\cdots A_M}$. The prover applies a multipartite entanglement-breaking channel on $R$, which can be written as
    \begin{equation}
        \mathcal{E}_{R\rightarrow A^{\prime}_1\cdots A^{\prime}_{M-1}}(\cdot)=\sum_{x}\operatorname{Tr}[\mu^{x}_{R}(\cdot)]\left(\phi^{x,1}_{A^{\prime}_1}\otimes\cdots\otimes\phi^{x,M-1}_{A^{\prime}_{M-1}}\right),
    \end{equation}
     where $\{\mu^{x}_{R}\}_{x}$ is a rank-one POVM and $\{\phi^{x,i}_{A^{\prime}_i}\}_{x,i}$ is a set of pure states. The prover sends the systems $(A^{M-1})^{\prime}= A^{\prime}_1\cdots A^{\prime}_{M-1}$ to the verifier. Now, the verifier performs a collective swap test on $A_1\cdots A_M$, as depicted in the final part of the circuit diagram in Figure~\ref{fig:Max_Sep_Fidelity_QIP_EB_Multi}. The acceptance probability of this interactive proof system is thus given by
    \begin{equation}\label{eqn:qipeb-mult}
        \max_{\mathcal{E}\in\text{EB}}\operatorname{Tr}[\Pi_{(A_1\cdots A_{M-1})^{\prime}A_1\cdots A_{M-1}}^{\text{sym}}\mathcal{E}_{R\rightarrow (A_1\cdots A_{M-1})^{\prime}}(\psi_{RA_1\cdots A_{M}})],
    \end{equation}
    where
    \begin{equation}
    \Pi^{\text{sym}}_{(A_1\cdots A_{M-1})^{\prime}A_1\cdots A_{M-1}}:=\frac{1}{2}\left(I_{(A_1\cdots A_{M-1})^{\prime}A_1\cdots A_{M-1}}+F_{(A_1\cdots A_{M-1})^{\prime}A_1\cdots A_{M-1}}\right)
    \end{equation}
    is the projector onto the symmetric subspace of $A^{\prime}$ and $A$ and $F_{(A_1\cdots A_{M-1})^{\prime}A_1\cdots A_{M-1}}$ is a tensor product of individual swaps $F_{A'_i A_i}$, for $i\in\{1,\ldots,M-1\}$. That is,
    \begin{equation}
        F_{(A_1\cdots A_{M-1})^{\prime}A_1\cdots A_{M-1}} = \bigotimes_{i=1}^{M-1} F_{A'_i A_i}.
    \end{equation}
    Then we find, for fixed $\mathcal{E}_{R\rightarrow A^{\prime}_1\cdots A^{\prime}_{M-1}}$, that
    \begin{align}
        &\operatorname{Tr}[\Pi_{(A_1\cdots A_{M-1})^{\prime}(A_1\cdots A_{M-1})}^{\text{sym}}\mathcal{E}_{R\rightarrow A^{\prime}_1\cdots A^{\prime}_{M-1}}(\psi_{RA_1\cdots A_{M}})]\notag \\
        &=\frac{1}{2}\operatorname{Tr}[(I_{(A_1\cdots A_{M-1})^{\prime}(A_1\cdots A_{M-1})}+F_{(A_1\cdots A_{M-1})^{\prime}(A_1\cdots A_{M-1})})
        \mathcal{E}_{R\rightarrow A^{\prime}_1\cdots A^{\prime}_{M-1}}(\psi_{RA_1\cdots A_{M}})]\\
        &=\frac{1}{2}+\frac{1}{2}\operatorname{Tr}[F_{(A_1\cdots A_{M-1})^{\prime}(A_1\cdots A_{M-1})}\mathcal{E}_{R\rightarrow A^{\prime}_1\cdots A^{\prime}_{M-1}}(\psi_{RA^{M}})]\\
        &=\frac{1}{2}+\frac{1}{2}\operatorname{Tr}\!\left[F_{(A_1\cdots A_{M-1})^{\prime}(A_1\cdots A_{M-1})}\sum_{x}\operatorname{Tr}[\mu^{x}_{R}(\psi_{RA^{M}})]   \phi^{x,1}_{A^{\prime}_1}\otimes\cdots\otimes\phi^{x,M-1}_{A^{\prime}_{M-1}}\right],\\
        &=\frac{1}{2}+\frac{1}{2}\operatorname{Tr}\!\left[F_{(A_1\cdots A_{M-1})^{\prime}(A_1\cdots A_{M-1})}\sum_{x}p(x)(\psi^{x}_{A_1\cdots A_{M}})\phi^{x,1}_{A^{\prime}_1}\otimes\cdots\otimes\phi^{x,M-1}_{A^{\prime}_{M-1}}\right],\\
        &=\frac{1}{2}+\frac{1}{2}\sum_{x}p(x)\operatorname{Tr}\!\left[\left(\phi^{x,1}_{A_1}\otimes\cdots\otimes \phi^{x,M-1}_{A_{M-1}}\right)\psi^{x}_{A_1\cdots A_{M-1}}\right]\label{eqn:mult-fid-sep-intermediate},
    \end{align}
    where
    \begin{align}
        p(x) &  \coloneqq \operatorname{Tr}[\mu^{x}_{R} \psi_{RA^{M}} ],\\
        \psi^{x}_{A_1\cdots A_{M}} &  \coloneqq \frac{1}{p(x)}\operatorname{Tr}_R[\mu^{x}_{R}\psi_{RA^{M}}].
    \end{align}
    For a given $x$, let us  simplify $F_s(\varphi_{A_1\cdots A_M})$ as defined in~\eqref{eqn:multi_fid_sep}, 
   \begin{align} 
        F_s(\varphi_{A_1\cdots A_M})&=\max_{\left\{|\phi^{i}\rangle_{A_i}\right\}_{i=1}^M}\left\vert\langle\varphi|_{A_1\cdots A_M}|\phi^{1}\rangle_{A_1}\otimes\cdots\otimes |\phi^{M}\rangle_{A_M}\right\vert^2 \\
        &=\max_{\left\{|\phi^{i}\rangle_{A_i}\right\}_{i=1}^M}\left\vert\langle\phi^{1}|_{A_1}\otimes\cdots\otimes \langle\phi^{M}|_{A_M}|\varphi\rangle_{A_1\cdots A_M}\right\vert^2\\
        &=\max_{\left\{|\phi^{i}\rangle_{A_i}\right\}_{i=1}^{M-1}}\left\Vert \langle\phi^{1}|_{A_1}\otimes\cdots\otimes\langle\phi^{M-1}|_{A_{M-1}}\otimes I_{A_M}|\varphi\rangle_{A^M}\right\Vert _{2}^{2}\\
        &=\max_{\left\{|\phi^{i}\rangle_{A_i}\right\}_{i=1}^{M-1}}\operatorname{Tr}[(|\phi^{1}\rangle\!\langle\phi^{1}|_{A_1}\otimes\cdots \otimes|\phi^{M-1}\rangle\!\langle\phi^{M-1}|_{A_{M-1}}\otimes I_{A_M})\varphi_{A_1\cdots A_M}]\\
        &=\max_{\left\{|\phi^{i}\rangle_{A_i}\right\}_{i=1}^{M-1}}\operatorname{Tr}[(|\phi^{1}\rangle\!\langle\phi^{1}|_{A_1}\otimes\cdots \otimes|\phi^{M-1}\rangle\!\langle\phi^{M-1}|_{A_{M-1}})\varphi_{A_1\cdots A_{M-1}}]\label{eqn:mult-fid-sep-pure}.
    \end{align}
    The first two equalities are from the definition and a rewriting. The third equality follows from the variational characterization of the Euclidean norm of a vector. Noting the form in~\eqref{eqn:mult-fid-sep-pure} and applying the maximization over entanglement-breaking channels of the form described in~\eqref{eqn:ent-break-mult} to~\eqref{eqn:mult-fid-sep-intermediate}, we arrive at the desired claim in~\eqref{eq:mult-part-fid-sep-test-acc-prob}.
    \end{proof}

    \begin{figure}
        \centering
        \includegraphics[width=0.7\columnwidth]{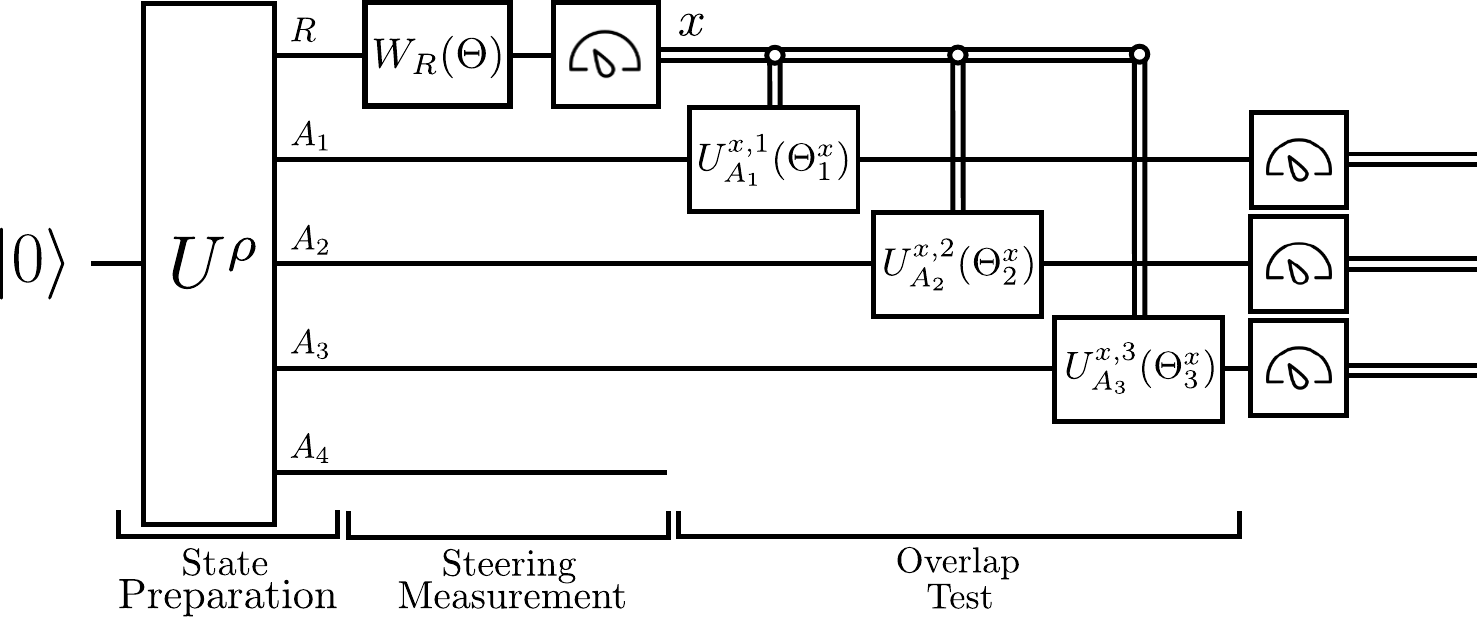}
        \caption{VQSA to estimate the multipartite fidelity of separability $F_s(\rho_{A_1A_2A_3A_4})$. 
        The unitary circuit $U^\rho$ prepares the state $\psi_{RA_1A_2A_3A_4}$, which is a purification of $\rho_{A_1A_2A_3A_4}$. The parameterized circuit $W_R(\Theta)$ acts on $R$ to evolve the state to another purification of $\rho_{A_1A_2A_3A_4}$. The following measurement, labeled ``steering measurement,'' steers the remaining systems to be in a state $\psi_{A_1A_2A_3A_4}^{x}$ if the measurement outcome $x$ occurs. Conditioned on the outcome $x$, the final parameterized circuits $U^{x,1}_{A_1}(\Theta^x_1)$, $U^{x,2}_{A_2}(\Theta^x_2)$, and $U^{x,3}_{A_3}(\Theta^x_3)$ are applied and the subsequent measurement estimates the quantity~$\operatorname{Tr}\!\left[\left(\phi^{x,1}_{A_1}\otimes\phi^{x,2}_{A_2}\otimes \phi^{x,3}_{A_{3}}\right)\psi^{x}_{A_1 A_2 A_{3}}\right]$.} 
        \label{fig:Max_Sep_Fidelity_VQSA_Multi}
    \end{figure}
    
    \bigskip 
    We can then use the generalized test of separability of mixed states to develop a VQSA for the multipartite case. See Figure~\ref{fig:Max_Sep_Fidelity_VQSA_Multi}. This involves replacing the collective swap test in Figure~\ref{fig:Max_Sep_Fidelity_QIP_EB_Multi} with an overlap measurement, similar to how we got Figure~\ref{fig:Max_Sep_Fidelity_VQA} in the main text from Figure~\ref{fig:Max_Sep_Fidelity_QIP_EB} in the main text.
    
\section{Complexity Class \texorpdfstring{\qipeb}{Lg}} \label{appendix:qipeb}

    In this appendix, we establish a complete problem for \qipeb, and then we interpret this problem in Remark~\ref{rem:interp-QIP-EB2}. See~\cite{watrous2009complexity, VW15} for further background on quantum computational complexity theory. Let us first define the complexity class \qipeb. Let $A=\left(A_{\text{yes}},A_{\text{no}}\right)$ be a promise problem, and let $a,b:\mathbb{N}\to [0,1]$ and $p$ be polynomial functions. The verifier $V$ is described by a polynomial-time generated family of quantum circuits. The prover $P$ is a family of arbitrary entanglement-breaking channels that naturally interface with a given verifier. Then $A\in$ \qipeb$(a,b)$ if there exists a two-message verifier with the following properties:
    \begin{enumerate}
        \item Completeness: For all $x\in A_{\text{yes}}$, there exists a prover $P$ that causes the verifier $V$ to accept $x$ with probability at least $a(|x|)$.
        \item Soundness: For all $x\in A_{\text{no}}$, every prover $P$  causes the verifier $V$ to accept $x$ with probability at most $b(|x|)$.
    \end{enumerate}
    In the above, acceptance is defined as obtaining the outcome one upon measuring the decision-qubit register.

\begin{problem}\label{prob:eb-2}
Given are circuits to generate a channel $\mathcal{N}_{G\rightarrow S}$ and a state $\rho_{S}$. Fix $\alpha$ and $\beta$ such that $0 \leq \alpha < \beta \leq 1$. Decide which of the following holds:
\begin{align}
    \text{Yes:} \quad f(\mathcal{N}_{G\rightarrow S},\rho_{S}) & \geq \beta, \\
    \text{No:} \quad f(\mathcal{N}_{G\rightarrow S},\rho_{S}) & \leq \alpha,
\end{align}
where
        \begin{equation}
            f(\mathcal{N}_{G\rightarrow S},\rho_{S}) \coloneqq  \max_{\substack{\{  (p(x),\psi^{x})\}  _{x},
            \left\{  \varphi^{x}\right\}_{x}
            }} \left\{  \sum_{x}p(x)F(\psi_{S}^{x},\mathcal{N}_{G\rightarrow S}(\varphi_{G}^{x})) : \sum_{x}p(x)\psi_{S}^{x}=\rho_{S}  \right\}
            \label{eq:accept-prob-qip_eb-2}
        \end{equation}
        with the optimization being over every pure-state decomposition of $\rho_{S}$ as $\sum_{x}p(x)\psi_{S}^{x}=\rho_{S}$. Also, $\left\{  \varphi^{x}\right\}_{x}$ is a set of pure states.
\end{problem}
    
    \begin{theorem}
         Problem~\ref{prob:eb-2} is a complete problem for $\operatorname{QIP}_{\operatorname{EB}}(2)$.
    \end{theorem}

    \begin{proof}
    The main idea behind the proof is to show that the acceptance probability of a general \qipeb\ problem can precisely be written as $f(\mathcal{N}_{G\rightarrow S},\rho_{S})$. This implies that an arbitrary \qipeb\ problem can be reduced to an instance of Problem~\ref{prob:eb-2}, and we argue at the end how this also implies that Problem~\ref{prob:eb-2} can be reduced to an instance of a problem in \qipeb.

    Consider a general interactive proof system in \qipeb\ that begins with the verifier preparing a bipartite pure state $\psi_{RS}$, followed by the system $R$ being sent to the prover, which subsequently performs an entanglement-breaking channel. The verifier then performs a unitary $V_{R^{\prime}S\rightarrow DG}$ and projects onto the $|1\rangle\!\langle 1|$ state of the decision qubit. Indeed, the acceptance probability is given by
    \begin{equation}
        \max_{\mathcal{E}\in\operatorname{EB}}\operatorname{Tr}[(|1\rangle\!\langle1|_{D}\otimes I_{G})\mathcal{V}_{R^{\prime}S\rightarrow DG}(\mathcal{E}_{R\rightarrow R^{\prime}}(\psi_{RS}))],
    \end{equation}
    where $\mathcal{V}_{R^{\prime}S\rightarrow DG}$ is the unitary channel corresponding to the unitary operator $V_{R^{\prime}S\rightarrow DG}$.
    By reasoning similar to that in~\eqref{eq:eb-proof-1},~\eqref{eq:eb-proof-2}, and~\eqref{eq:eb-proof-3}, we find that
    \begin{equation}
        \mathcal{E}_{R\rightarrow R^{\prime}}(\psi_{RS})=\sum_{x}p(x)\phi_{R^{\prime}}^{x}\otimes\psi_{S}^{x},
    \end{equation}
    so that the acceptance probability is equal to
    \begin{multline}
          \max_{\substack{\{  (p(x),\psi^{x})\}  _{x},\\
          \left\{  \phi^{x}\right\}_{x},\\
          \sum_{x}p(x)\psi_{S}^{x}=\psi_{S}}}
            \operatorname{Tr}\!\left[  (|1\rangle\!\langle1|_{D}\otimes I_{G})\mathcal{V}\left(  \sum_{x}p(x)\phi_{R^{\prime}}^{x}\otimes\psi_{S}^{x}\right)  \right]
          \\
          =\max_{\substack{\{  (p(x),\psi^{x})\}  _{x},\\
          \left\{  \phi^{x}\right\}_{x},\\
          \sum_{x}p(x)\psi_{S}^{x}=\psi_{S}}}
            \sum_{x}p(x)\operatorname{Tr}\!\left[  (|1\rangle\!\langle 1|_{D}\otimes I_{G})\mathcal{V}
            \left(\phi_{R^{\prime}}^{x}\otimes\psi_{S}^{x}\right)  \right] ,
    \end{multline}
    where we have used the shorthand $\mathcal{V} \equiv \mathcal{V}_{R^{\prime}S\rightarrow DG}$.
    Consider that, for all $x$,
    \begin{align}
        \operatorname{Tr}\!\left[  (|1\rangle\!\langle1|_{D}\otimes I_{G})\mathcal{V} \left(  \phi_{R^{\prime}}^{x}\otimes\psi_{S}^{x}\right)  \right] 
        &=\left\Vert \langle1|_{D}\otimes I_{G}) V|\phi^{x}\rangle_{R^{\prime}}\otimes|\psi^{x}\rangle_{S}\right\Vert _{2}^{2}\\
        &=\max_{|\varphi^{x}\rangle_{G}}\left\vert \langle1|_{D}\otimes\langle\varphi^{x}|_{G}) V |\phi^{x}\rangle_{R^{\prime}}\otimes|\psi^{x}\rangle_{S}\right\vert ^{2}\\
        &  =\max_{|\varphi^{x}\rangle_{G}}\operatorname{Tr}\!\left[  V^{\dag}(|1\rangle\!\langle1|_{D}\otimes|\varphi^{x}\rangle\!\langle\varphi^{x}|_{G})  V\phi^{x}_{R^{\prime}}\otimes|\psi^{x}\rangle\!\langle\psi^{x}|_{S}\right]  \\
        &  =\max_{|\varphi^{x}\rangle_{G}}\operatorname{Tr}\!\left[  \mathcal{W}_{G\rightarrow R^{\prime}S}(|\varphi^{x}\rangle\!\langle\varphi^{x}|_{G})\phi^{x}_{R^{\prime}}\otimes|\psi^{x}\rangle\!\langle\psi^{x}|_{S}\right],
    \end{align}
    where the isometric channel $\mathcal{W}_{G \to R'S}$ is defined as 
    \begin{equation}
        \mathcal{W}_{G\rightarrow R^{\prime}S}(\cdot)\coloneqq (V_{R^{\prime}S\rightarrow DG})^{\dag}(|1\rangle\!\langle1|_{D}\otimes(\cdot)_{G})V_{R^{\prime}S\rightarrow DG},
    \end{equation}
    and the corresponding isometry $W_{G \to R'S}$ as $W_{G \to R'S} \coloneqq (V_{R^{\prime}S\rightarrow DG})^{\dag}|1\rangle_{D}$.
    Then, the acceptance probability is given by
    \begin{equation}
    \max_{\substack{\{  (p(x),\psi^{x})\}  _{x},\\\left\{\phi^{x}\right\}  _{x},\left\{  \varphi^{x}\right\}  _{x}}}
    \left\{\begin{array}
            [c]{c}%
            \sum_{x}p(x)\operatorname{Tr}\!\left[\mathcal{W}_{G\rightarrow R^{\prime}S}(|\varphi^{x}\rangle\!\langle\varphi^{x}|_{G})\phi^{x}_{R^{\prime}}\otimes|\psi^{x}\rangle\!\langle\psi^{x}|_{S}\right]  :\\
            \sum_{x}p(x)\psi_{S}^{x}=\psi_{S}%
        \end{array}\right\}  .
    \end{equation}
    Since the optimization over $\phi^{x}_{R^{\prime}}$ is arbitrary, we can also write%
    \begin{align}
        &\max_{|\phi^{x}\rangle_{R^{\prime}}}\operatorname{Tr}\!\left[  \mathcal{W}_{G\rightarrow R^{\prime}S}(|\varphi^{x}\rangle\!\langle\varphi^{x}|_{G})\phi^{x}_{R^{\prime}}\otimes|\psi^{x}\rangle\!\langle\psi^{x}|_{S}\right] \notag \\
        &= \max_{|\phi^{x}\rangle_{R^{\prime}}}\left\vert \langle\phi^{x}|_{R^{\prime}}\otimes\langle\psi^{x}|_{S}W_{G\rightarrow R^{\prime}S}|\varphi^{x}\rangle_{G}\right\vert ^{2} \\
        &  =\left\Vert I_{R^{\prime}}\otimes\langle\psi^{x}|_{S}W_{G\rightarrow R^{\prime}S}|\varphi^{x}\rangle_{G}\right\Vert _{2}^{2} \\
        &  =\left(  \langle\varphi^{x}|_{G}\left(  W_{G\rightarrow R^{\prime}S}\right)  ^{\dag}I_{R^{\prime}}\otimes|\psi^{x}\rangle_{S}\right)   \left(I_{R^{\prime}}\otimes\langle\psi^{x}|_{S} W_{G\rightarrow R^{\prime}S}|\varphi^{x}\rangle_{G}\right)  \\
        &  =\langle\varphi^{x}|_{G}\left(W_{G\rightarrow R^{\prime}S}\right)^{\dag}\left(  I_{R^{\prime}}\otimes|\psi^{x}\rangle\!\langle\psi^{x}|_{S}\right)  W_{G\rightarrow R^{\prime}S}|\varphi^{x}\rangle_{G}\\
        &  =\operatorname{Tr}\!\left[  \left(  I_{R^{\prime}}\otimes|\psi^{x}\rangle\!\langle\psi^{x}|_{S}\right)  W_{G\rightarrow R^{\prime}S}|\varphi^{x}\rangle\!\langle\varphi^{x}|_{G}\left(  W_{G\rightarrow R^{\prime}S}\right)^{\dag}\right]  \\
        &  =\operatorname{Tr}[|\psi^{x}\rangle\!\langle\psi^{x}|_{S}\mathcal{N}_{G\rightarrow S}(|\varphi^{x}\rangle\!\langle\varphi^{x}|_{G})],
    \end{align}
    where we define the channel $\mathcal{N}_{G\rightarrow S}$ as
    \begin{equation}
        \mathcal{N}_{G\rightarrow S}(\cdot)\coloneqq \operatorname{Tr}_{R^{\prime}}[(V_{R^{\prime}S\rightarrow DG})^{\dag}(|1\rangle\!\langle1|_{D}\otimes(\cdot)_{G})V_{R^{\prime}S\rightarrow DG}].
        \label{eq:eb2-channel-def}
    \end{equation}
    Then, we find that the acceptance probability is given by
    \begin{equation}
          \max_{\substack{\{  (p(x),\psi^{x})\}  _{x},\\
          \left\{  \varphi^{x}\right\}_{x},\\
          \sum_{x}p(x)\psi_{S}^{x}=\psi_{S}}}  \sum_{x}p(x)\operatorname{Tr}[|\psi^{x}\rangle\!\langle\psi^{x}|_{S}\mathcal{N}_{G\rightarrow S}(|\varphi^{x}\rangle\!\langle\varphi^{x}|_{G})]
         =\max_{\substack{\{  (p(x),\psi^{x})\}  _{x},\left\{  \varphi^{x}\right\}_{x}\\
         \sum_{x}p(x)\psi_{S}^{x}=\psi_{S}
         }}  \sum_{x}p(x)F(\psi_{S}^{x},\mathcal{N}_{G\rightarrow S}(\varphi_{G}^{x}))  .
    \end{equation}
    This concludes the proof of the first part. 

    To see how this implies that Problem~\ref{prob:eb-2} can be realized in \qipeb, note that the circuit preparing the state $\rho_S$ prepares a purification and traces over the reference system, and the circuit to generate $\mathcal{N}_{G\to S}$ is realize by adjoining an environment system in the state $|0\rangle\!\langle 0|$, performing a unitary, and tracing over the environment. So we let the verifier prepare the purification of $\rho_S$ and this plays the role of $\psi_{RS}$ above, and the channel $\mathcal{N}_{G\to S}$ can be realized precisely as in~\eqref{eq:eb2-channel-def} with appropriate substitutions.
    \end{proof}

    \begin{remark}
    \label{rem:interp-QIP-EB2}
    The quantity in~\eqref{eq:accept-prob-qip_eb-2} can be interpreted as follows: Given a channel $\mathcal{N}$ and a source state~$\rho$, calculate the largest average ensemble fidelity attainable in reproducing the source at the output of the channel. This means it is necessary to find the ensemble decomposition $\{  (p(x),\psi^{x})\}  _{x}$ of $\rho$ as well as a set $\left\{  \varphi^{x}\right\}_{x}$ of encoding states 
    that lead to the largest ensemble fidelity (and this is what is left to the prover). This criterion is similar to one used in Schumacher data compression~\cite{PhysRevA.51.2738}, but this seems more similar to the setting of the source-channel separation theorem~\cite{6287590}, in which the goal is to transmit an information source over a quantum channel. The channel $\mathcal{N}$ here could consist of a fixed encoding $\mathcal{E}$, noisy channel $\mathcal{M}$, and fixed decoding $\mathcal{D}$, (i.e., $\mathcal{N} = \mathcal{D} \circ \mathcal{M} \circ \mathcal{E}$) and then the goal is to test how well a given fixed scheme $(\mathcal{E},\mathcal{D})$ can communicate a source $\rho$ over a channel $\mathcal{M}$, according to the ensemble fidelity criterion.
    \end{remark}

    \begin{remark}
    \label{rem:concave-closure-QIPEB2}
        We can write the expression in~\eqref{eq:accept-prob-qip_eb-2} alternatively as
        \begin{equation}
            \text{\rm Eq.~\eqref{eq:accept-prob-qip_eb-2}} = \max_{\substack{\{  (p(x),\psi^{x})\}  _{x},\\
            \sum_{x}p(x)\psi_{S}^{x}=\rho_{S}
            }}   \sum_{x}p(x)\left\|(\mathcal{N}_{G\rightarrow S})^{\dag}(\psi_{S}^{x})\right\|_\infty,
            \label{eq:concave-closure}
        \end{equation}
        where $(\mathcal{N}_{G\rightarrow S})^{\dag}$ is the Hilbert--Schmidt adjoint of the channel $(\mathcal{N}_{G\rightarrow S})^{\dag}$.
        Employing the abbreviations $\psi^{x}_{S}\equiv |\psi^{x}\rangle\!\langle\psi^{x}|_{S}$ and $\varphi^{x}_{G}\equiv |\varphi^{x}\rangle\!\langle\varphi^{x}|_{G}$, this follows because
        \begin{align}
        & \max_{\substack{\{  (p(x),\psi^{x})\}  _{x},\\
          \left\{  \varphi^{x}\right\}_{x},\\
          \sum_{x}p(x)\psi_{S}^{x}=\psi_{S}}}  \sum_{x}p(x)\operatorname{Tr}[\psi^{x}_{S} \mathcal{N}_{G\rightarrow S}(\varphi^{x}_{G})] \notag \\
          & = \max_{\substack{\{  (p(x),\psi^{x})\}  _{x},\\
          \left\{  \varphi^{x}\right\}_{x},\\
          \sum_{x}p(x)\psi_{S}^{x}=\psi_{S}}}  \sum_{x}p(x)\operatorname{Tr}[(\mathcal{N}_{G\rightarrow S})^{\dag}(\psi^{x}_{S})\varphi^{x}_{G}] \\
          & = \max_{\substack{\{  (p(x),\psi^{x})\}  _{x},
          \\
          \left\{  \varphi^{x}\right\}_{x},\\
          \sum_{x}p(x)\psi_{S}^{x}=\psi_{S}}}  \sum_{x}p(x)
          \max_{\left\{  \varphi^{x}\right\}_{x}}
          \operatorname{Tr}[(\mathcal{N}_{G\rightarrow S})^{\dag}(\psi^{x}_{S})\varphi^{x}_{G}] \\
          & = \max_{\substack{\{  (p(x),\psi^{x})\}  _{x},\\
            \sum_{x}p(x)\psi_{S}^{x}=\rho_{S}
            }}   \sum_{x}p(x)\left\|(\mathcal{N}_{G\rightarrow S})^{\dag}(\psi_{S}^{x})\right\|_\infty.
        \end{align}
        If we define the function
        \begin{equation}
            g_{\mathcal{N}}(\rho) \coloneqq \left\|(\mathcal{N}_{G\rightarrow S})^{\dag}(\rho_S)\right\|_\infty,
        \end{equation}
        then the function in~\eqref{eq:concave-closure} is known as the concave closure of $g_{\mathcal{N}}(\rho)$ and has been studied in other contexts in quantum information theory \cite[Section~2]{AFKW04}. It has an interesting dual formulation, as demonstrated in \cite[Eq.~(15)]{AFKW04}. Given the observation in~\eqref{eq:concave-closure}, we can thus conclude that, given circuits to realize the channel $\mathcal{N}$ and state $\rho$, estimating the concave closure of the function $g_{\mathcal{N}}(\rho) $ within additive error is a complete problem for {\rm \qipeb}.
    \end{remark}

    \begin{remark}
        Employing the reasoning from Remark~\ref{rem:concave-closure-QIPEB2}, we find that the acceptance probability in~\eqref{eqn:accept_prob_supp} is equal to the concave closure of the following function:
        \begin{equation}
            f(\rho_{AB}) \coloneqq \left \| \Pi^{\operatorname{sym}}_{AA'} (\rho_{AB} \otimes I_{A'}) \Pi^{\operatorname{sym}}_{AA'}\right \|_{\infty},
            \label{eq:concave-closure-sep-prob-accept-prob}
        \end{equation}
        where we used the fact that the state $\rho_S$ from Remark~\ref{rem:concave-closure-QIPEB2} is $\rho_{AB}$ and the map $\mathcal{N}_{G \to S}$ from Remark~\ref{rem:concave-closure-QIPEB2} is
        \begin{equation}
         \mathcal{N}(\sigma_{AA'B}) = \operatorname{Tr}_{A'}[\Pi^{\operatorname{sym}}_{AA'} \sigma_{AA'B}\Pi^{\operatorname{sym}}_{AA'}],   
        \end{equation}
          with adjoint
          \begin{equation}
          \mathcal{N}^{\dag}(\omega_{AB}) = \Pi^{\operatorname{sym}}_{AA'} (\omega_{AB} \otimes I_{A'}) \Pi^{\operatorname{sym}}_{AA'}.    
          \end{equation}
        Observe that the map $\rho_{AB} \mapsto \Pi^{\operatorname{sym}}_{AA'} (\rho_{AB} \otimes I_{A'}) \Pi^{\operatorname{sym}}_{AA'}$ is proportional to that used in a $1\to 2$ universal cloning machine \cite[Eq.~(17)]{SIGA05}.
        If $\rho_{AB}$ is pure, so that we write it as $\psi_{AB}$, then the following inequality holds:
        \begin{equation}
            \left \| \Pi^{\operatorname{sym}}_{AA'} (\psi_{AB} \otimes I_{A'}) \Pi^{\operatorname{sym}}_{AA'}\right \|_{\infty}
             \leq \left \| \Pi^{\operatorname{sym}}_{AA'} \right \|_{\infty}
            \left \|
            \psi_{AB} \otimes I_{A'} \right \|_{\infty}
            \left \|\Pi^{\operatorname{sym}}_{AA'}
            \right \|_{\infty}
             \leq 1,
        \end{equation}
        where we applied the multiplicativity of the spectral norm. 
        Thus, the concave closure of $f(\rho_{AB})$ satisfies $f(\rho_{AB}) \in [0,1]$. Furthermore, from Lemma~\ref{lem:alt-accept-prob-pure-states} below, we know that
        \begin{equation}
            \left \| \Pi^{\operatorname{sym}}_{AA'} (\psi_{AB} \otimes I_{A'}) \Pi^{\operatorname{sym}}_{AA'}\right \|_{\infty} 
        = 
            \frac{1}{2}\left(  1+\left\Vert \psi_{A} \right\Vert _{\infty}\right),
        \end{equation}
        showing the consistency of the claim just above~\eqref{eq:concave-closure-sep-prob-accept-prob} with Theorem~\ref{theorem:qip-eb_msf} and Eqs.~\eqref{eq:convex-decomp-max-sep-fid} and~\eqref{eq:max-sep-fid-inf-norm} in the main text.
        If $\rho_{AB}$ is a pure product state, so that we can write it as $\rho_{AB} = \phi_A \otimes \varphi_B$, then we have that
        \begin{equation}
            \left \| \Pi^{\operatorname{sym}}_{AA'} (\phi_A \otimes \varphi_B \otimes I_{A'}) \Pi^{\operatorname{sym}}_{AA'}\right \|_{\infty} = 
            \left \| \Pi^{\operatorname{sym}}_{AA'} (\phi_A  \otimes I_{A'}) \Pi^{\operatorname{sym}}_{AA'}\right \|_{\infty},
            \label{eq:spec-norm-analysis}
        \end{equation}
        and the spectral norm on the right-hand side of~\eqref{eq:spec-norm-analysis} is achieved by choosing the vector $|\phi\rangle_A \otimes |\phi\rangle_{A'}$, so that
        \begin{align}
            & (\langle\phi|_A \otimes \langle\phi|_{A'})\Pi^{\operatorname{sym}}_{AA'} (\phi_A  \otimes I_{A'}) \Pi^{\operatorname{sym}}_{AA'}(|\phi\rangle_A \otimes |\phi\rangle_{A'}) \notag 
            \\
            &  = (\langle\phi|_A \otimes \langle\phi|_{A'}) (\phi_A  \otimes I_{A'}) (|\phi\rangle_A \otimes |\phi\rangle_{A'})\\
            &  =1.
        \end{align}
    \end{remark}

    \begin{lemma}
    \label{lem:alt-accept-prob-pure-states}
        For a pure state $\psi_{AB}$, the following equality holds:
        \begin{equation}
        \left \| \Pi^{\operatorname{sym}}_{AA'} (\psi_{AB} \otimes I_{A'}) \Pi^{\operatorname{sym}}_{AA'}\right \|_{\infty} 
        = 
            \frac{1}{2}\left(  1+\left\Vert \psi_{A} \right\Vert _{\infty}\right),
        \end{equation}
        where $\psi_{A}\equiv \operatorname{Tr}_B[\psi_{AB}]$.
    \end{lemma}

    \begin{proof}
Consider that
\begin{multline}
\left\Vert \left(  \Pi_{AA^{\prime}}^{\text{sym}}\otimes I_{B}\right)  \left(
|\psi\rangle\!\langle\psi|_{AB}\otimes I_{A^{\prime}}\right)  \left(
\Pi_{AA^{\prime}}^{\text{sym}}\otimes I_{B}\right)  \right\Vert _{\infty}
\\
=\left\Vert \left(  |\psi\rangle\!\langle\psi|_{AB}\otimes I_{A^{\prime}%
}\right)  \left(  \Pi_{AA^{\prime}}^{\text{sym}}\otimes I_{B}\right)  \left(
|\psi\rangle\!\langle\psi|_{AB}\otimes I_{A^{\prime}}\right)  \right\Vert
_{\infty}.%
\end{multline}
Now consider that
\begin{align}
& \left(  |\psi\rangle\!\langle\psi|_{AB}\otimes I_{A^{\prime}}\right)  \left(
\Pi_{AA^{\prime}}^{\text{sym}}\otimes I_{B}\right)  \left(  |\psi
\rangle\!\langle\psi|_{AB}\otimes I_{A^{\prime}}\right)  \nonumber\\
& =\left(  |\psi\rangle\!\langle\psi|_{AB}\otimes I_{A^{\prime}}\right)  \left(
\frac{I_{AA^{\prime}}+F_{AA^{\prime}}}{2}\otimes I_{B}\right)  \left(
|\psi\rangle\!\langle\psi|_{AB}\otimes I_{A^{\prime}}\right)   \\
& =\frac{1}{2}\left(  |\psi\rangle\!\langle\psi|_{AB}\otimes I_{A^{\prime}%
}\right)+\frac{1}{2}\left(  |\psi\rangle\!\langle\psi|_{AB}\otimes
I_{A^{\prime}}\right)  \left(  F_{AA^{\prime}}\otimes I_{B}\right)  \left(
|\psi\rangle\!\langle\psi|_{AB}\otimes I_{A^{\prime}}\right)  .
\end{align}
Then writing the Schmidt decomposition of $|\psi\rangle_{AB}$ as $|\psi
\rangle_{AB}=\sum_{i}\sqrt{\lambda_{i}}|i\rangle_{A}|i\rangle_{B}$, we find that
\begin{align}
    & \left(  |\psi\rangle\!\langle\psi|_{AB}\otimes I_{A^{\prime}}\right)  \left(
    F_{AA^{\prime}}\otimes I_{B}\right)  \left(  |\psi\rangle\!\langle\psi
    |_{AB}\otimes I_{A^{\prime}}\right)  \nonumber\\
    &=\sum_{i,i^{\prime},j,j^{\prime}}\sqrt{\lambda_{i}\lambda_{i^{\prime}}%
    }\left(  |\psi\rangle\!\langle i|_{A}\langle i|_{B}\otimes|j\rangle\!\langle
    j|_{A^{\prime}}\right)  \left(  F_{AA^{\prime}}\otimes I_{B}\right) \left(
    |i^{\prime}\rangle_{A}|i^{\prime}\rangle_{B}\langle\psi|_{AB}\otimes
    |j^{\prime}\rangle\!\langle j^{\prime}|_{A^{\prime}}\right)  \\
    &=\sum_{i,i^{\prime},j,j^{\prime}}\sqrt{\lambda_{i}\lambda_{i^{\prime}}}\left(  |\psi\rangle\!\langle i|_{A}\langle i|_{B}\otimes|j\rangle\!\langle
    j|_{A^{\prime}}\right) \left(  |j^{\prime}\rangle_{A}|i^{\prime}\rangle_{B}\langle\psi|_{AB}\otimes|i^{\prime}\rangle\!\langle j^{\prime}|_{A^{\prime}%
    }\right)  \\
    &=\sum_{i,i^{\prime},j,j^{\prime}}\sqrt{\lambda_{i}\lambda_{i^{\prime}}}%
|\psi\rangle\!\langle i|j^{\prime}\rangle_{A}\langle i|i^{\prime}\rangle
_{B}\langle\psi|_{AB}\otimes|j\rangle\!\langle j|i^{\prime}\rangle\!\langle
j^{\prime}|_{A^{\prime}}\\
& =\sum_{i}\lambda_{i}|\psi\rangle\!\langle\psi|_{AB}\otimes|i\rangle\!\langle
i|_{A^{\prime}} \\
& =|\psi\rangle\!\langle\psi|_{AB}\otimes\sum_{i}\lambda_{i}|i\rangle\!\langle
i|_{A^{\prime}} \\ & =|\psi\rangle\!\langle\psi|_{AB}\otimes\psi_{A^{\prime}}.
\end{align}
Then%
\begin{align}
&\left(  |\psi\rangle\!\langle\psi|_{AB}\otimes I_{A^{\prime}}\right)  \left(
\Pi_{AA^{\prime}}^{\text{sym}}\otimes I_{B}\right)  \left(  |\psi
\rangle\!\langle\psi|_{AB}\otimes I_{A^{\prime}}\right)  \nonumber\\
& =\frac{1}{2}\left(  |\psi\rangle\!\langle\psi|_{AB}\otimes I_{A^{\prime}%
}\right)  +|\psi\rangle\!\langle\psi|_{AB}\otimes\frac{1}{2}\psi_{A^{\prime}}\\
& =|\psi\rangle\!\langle\psi|_{AB}\otimes\frac{1}{2}\left(  I_{A^{\prime}}%
+\psi_{A^{\prime}}\right)  ,
\end{align}
and we conclude that%
\begin{align}
& \left\Vert \left(  \Pi_{AA^{\prime}}^{\text{sym}}\otimes I_{B}\right)
\left(  |\psi\rangle\!\langle\psi|_{AB}\otimes I_{A^{\prime}}\right)  \left(
\Pi_{AA^{\prime}}^{\text{sym}}\otimes I_{B}\right)  \right\Vert _{\infty
}\nonumber\\
& =\left\Vert |\psi\rangle\!\langle\psi|_{AB}\otimes\frac{1}{2}\left(
I_{A^{\prime}}+\psi_{A^{\prime}}\right)  \right\Vert _{\infty}\\
& =\frac{1}{2}\left(  1+\left\Vert \psi_{A^{\prime}}\right\Vert _{\infty
}\right)  
 \\
 &=\frac{1}{2}\left(  1+\left\Vert \psi_{A}\right\Vert _{\infty}\right)  .
\end{align}
This concludes the proof.
    \end{proof}
    
\section{Placement of \texorpdfstring{\qipeb}{Lg}}
\label{appendix:complexity_placements}

    In this appendix, we establish the following containments:
    \begin{equation}
    \text{QAM}, \text{QSZK} \subseteq \text{\qipeb}.    
    \end{equation}
    See Figure~\ref{fig:placement} for a detailed diagram.

    \subsection{QAM \texorpdfstring{$\subseteq$}{Lg} \texorpdfstring{\qipeb}{Lg}}

    First, recall that QAM consists of the verifier selecting a classical letter $x$ uniformly at random, sending the choice to the prover, who then sends back a pure state $\psi_x$ to the verifier, who finally performs an efficient measurement to decide whether to accept the computation~\cite{marriott2004quantum}. Note that QAM contains QMA~\cite{marriott2004quantum}.
    
    To see the containment QAM $\subseteq$ \qipeb, consider that the verifier's first circuit in \qipeb\ can consist of preparing a random classical bitstring in a system $R$. The verifier sends system $R$ to the prover. Then, the prover's action amounts to preparing some state that gets returned to the verifier. The rest of the protocol then simulates a QAM protocol.

    \subsection{QSZK \texorpdfstring{$\subseteq$}{Lg} \texorpdfstring{\qipeb}{Lg}}

    Quantum statistical zero-knowledge (QSZK) consists of all problems that can be solved by the interaction between a quantum verifier and a quantum prover, such that the verifier accumulates statistical evidence about the answer to a decision, but does not learn anything other than the answer by interacting with the prover~\cite{W02, watrous2006zero}. A complete problem for this class is quantum state distinguishability, in which the goal is to decide whether two states $\rho_0$ and $\rho_1$, generated by quantum circuits, are far or close in trace distance~\cite{W02}. This is a nice problem for understanding the basics of the QSZK complexity class: the interaction begins with the verifier picking one of the states uniformly at random, recording the choice as a bit $x$, and then sending the chosen state~$\rho_x$ to the prover over a quantum channel. The prover can then perform the optimal Helstrom measurement~\cite{H67, H69} to distinguish the states, which has success probability equal to
    \begin{equation}
        p_{\text{succ}}\coloneqq \frac{1}{2}\left(1 + \frac{1}{2}\left \|  \rho_0 - \rho_1\right\|_1\right).
    \end{equation}
    The Helstrom measurement leads to a decision bit $y$, which the prover sends back to the verifier over a quantum channel (here, a single classical bit channel would suffice). The verifier then accepts if $x = y$, and the probability that this happens is equal to $p_{\text{succ}}$. By repeating this protocol a polynomial number of times and invoking the error-reduction protocol from~\cite{W02}, it follows that the verifier can make the completeness and soundness probabilities exponentially close to one and zero, respectively, to have essentially zero error probability in the final decision about whether the states are near or far in trace distance. Finally, the interaction has a ``zero knowledge'' aspect because the verifier only learns the bit of the prover and nothing about how to distinguish the states. 

    Since quantum state distinguishability is a complete problem for QSZK and the interaction described above can be performed in \qipeb, the containment QSZK $\subseteq$ \qipeb\ follows.

\section{Local Reward Function}
\label{appendix:local_cost}

    In this appendix, we develop a local reward function as an alternative to the global reward function considered in the main text, i.e., the acceptance probability in  Theorem~\ref{theorem:global_cost_func}. 
    The acceptance probability in Theorem~\ref{theorem:global_cost_func} can be considered a global reward function because it corresponds to the probability of measuring zero in every register. As indicated in~\cite{Cerezo2021a}, it is helpful to employ a local reward function to mitigate the barren plateau problem~\cite{McClean2018}, which plagues all variational quantum algorithms. 
    
    Let us define the local and global reward functions. Let $Z_i$ be the event of measuring zero in the $i$-th register. We then set the local reward function to be the probability of measuring zero in a register chosen uniformly at random; that is, it is given by the following:
    \begin{equation}
        L \equiv \frac{1}{n} \sum_i \Pr\!\left(Z_i\right).
    \end{equation}
    The event of measuring all zeros is given by $\bigcap_i Z_i$, and the probability that this event occurs is $G \equiv \Pr\!\left(\bigcap_i Z_i\right)$, which is what we used in the main text as the global reward function.
    
    We are interested in determining inequalities related to the global and local reward functions, and the following analysis employs the same ideas used in \cite[Appendix C]{Khatri2019quantumassisted}. Using DeMorgan's laws, we find that
    \begin{equation}
        \Pr\!\left(\bigcap_i Z_i\right) = \Pr\!\left(\left(\bigcup_i Z_i^c\right)^c\right) = 1- \Pr\!\left(\bigcup_i Z_i^c\right).
    \end{equation}
    We can then use the union bound to conclude that
    \begin{equation}
        \Pr\!\left(\bigcap_i Z_i\right) = 1- \Pr\!\left(\bigcup_i (Z_i)^c\right) \geq 1- \sum_i \Pr\!\left((Z_i)^c\right).
    \end{equation}
    Finally, consider that
    \begin{align}
    \label{eqn:lower_bnd_local_cost}
        G & = \Pr\!\left(\bigcap_i Z_i\right) \\
        & \geq 1- \sum_i \Pr\!\left(Z_i^c\right) \\
        & = \sum_i \Pr\!\left(Z_i\right) - (n-1)\\
        & = n L - (n-1)\\
        & = n (L-1) + 1.
    \end{align}

    We can also derive an upper bound on the global reward function in terms of the local reward function. Recall the following inequality, which holds for every set $\{A_1,A_2,\ldots,A_n\}$ of events:
    \begin{equation}
        \Pr\!\left(\bigcup_i A_i\right) \geq \frac{1}{n} \sum_i \Pr\!\left(A_i\right).
    \end{equation}
    Setting $A_i=Z_i^c$, we get
    \begin{equation}
        \Pr\!\left(\bigcup_i Z_i^c\right) \geq \frac{1}{n} \sum_i \Pr\!\left(Z_i^c\right).
    \end{equation}
    Using DeMorgan's laws, we obtain the desired upper bound as follows:
    \begin{align}
        G = \Pr\!\left(\bigcap_i Z_i\right) & \leq 1-\frac{1}{n} \sum_i \left(1-\Pr\!\left(Z_i\right)\right)\\
        &= \frac{1}{n} \sum_i \Pr\!\left(Z_i\right) = L.
    \end{align}

    In summary, we have established the following bounds:
    \begin{equation}
          n (L-1) + 1 \leq  G \leq L,
    \end{equation}
    so that $G = 1 $ if and only if $L=1$. Since we always have $G \in [ 0,1]$, the lower bound is only nontrivial if $L$ is sufficiently large, i.e., if $L \geq 1 - \frac{1}{n}$.    
\end{document}